\documentclass[a4paper,11pt]{article}

\usepackage[USenglish]{babel}

\usepackage[T1]{fontenc}
\usepackage[utf8]{inputenc}
\usepackage{url}

\usepackage[a4paper,margin=2cm]{geometry}
\usepackage{microtype}
\usepackage{enumitem}
\usepackage{amsmath,amssymb,amsfonts}
\usepackage{amsthm}
\usepackage{soul}
\usepackage{thmtools}      
\usepackage{thm-restate}   
\usepackage{multirow}
\usepackage{hyperref}
\usepackage[nameinlink,noabbrev]{cleveref} 

\usepackage{graphicx}
\usepackage[ruled,vlined,linesnumbered]{algorithm2e}
\usepackage{xspace}
\usepackage{bbding}

\usepackage{xcolor}
\usepackage[colorinlistoftodos,prependcaption,textsize=tiny]{todonotes}

\usepackage[numbers]{natbib}

\newcommand{\leader}{\texttt{Leader}\xspace}
\newcommand{\nonleader}{\texttt{Non-Leader}\xspace}

\newcommand{\msg}{\texttt{m}\xspace}

\newcommand{\myUB}{\textsf{d}\xspace}

\newcommand{\relay}{{KillACWMessageAndRelay}\xspace}

\newcommand{\actif}{A\xspace}

\newcommand{\lost}{B\xspace}

\newcommand{\cA}{\ensuremath{\mathcal{A}\xspace}}

\newcommand{\N}{\ensuremath{\mathbb{N}\xspace}}

\newcommand{\myID}{\ensuremath{\mathsf{ID}\xspace}}

\newcommand{\id}{\mathsf{ID}\xspace}
\newcommand{\rand}{\mathsf{rand}\xspace}

\newcommand{\idmax}{\id_{\max}}
\newcommand{\idmin}{\id_{\min}}

\newcommand{\cw}{\text{CW}}
\newcommand{\ccw}{\text{CCW}}

\DeclareMathOperator{\len}{len}
\DeclareMathOperator{\bit}{bit}
\DeclareMathOperator{\sol}{sol}
\DeclareMathOperator{\pred}{pred}
\DeclareMathOperator{\pref}{pref}
\DeclareMathOperator{\suc}{succ}

\newtheorem{lemma}{Lemma}
\newtheorem{proposition}{Proposition}

\newtheorem{observation}{Observation}
\newtheorem{claim}{Claim}
\newtheorem{definition}{Definition}

\title{Non-Uniform Content-Oblivious Leader \\ Election on Oriented Asynchronous Rings}

\author{
  \hspace{1cm} Jérémie Chalopin\thanks{Aix Marseille Univ, CNRS, LIS, Marseille, France. \texttt{jeremie.chalopin@lis-lab.fr}. ORCID: \href{https://orcid.org/0000-0002-2988-8969}{0000-0002-2988-8969}.}
  \and Yi-Jun Chang\thanks{National University of Singapore, Singapore. \texttt{cyijun@nus.edu.sg}. ORCID: \href{https://orcid.org/0000-0002-0109-2432}{0000-0002-0109-2432}.}
  \and Lyuting Chen\thanks{National University of Singapore, Singapore. \texttt{e0726582@u.nus.edu.sg}. ORCID: \href{https://orcid.org/0009-0002-8836-6607}{0009-0002-8836-6607}.}  \hspace{1cm}
  \and Giuseppe A. Di Luna\thanks{DIAG, Sapienza University of Rome, Italy. \texttt{diluna@diag.uniroma1.it}.}
  \and Haoran Zhou\thanks{National University of Singapore, Singapore. \texttt{haoranz@u.nus.edu}. ORCID: \href{https://orcid.org/0009-0001-2458-5344}{0009-0001-2458-5344}.}
}

\date{} 

\begin{document}

\maketitle










\begin{abstract}

We study the leader election problem in oriented ring networks under content-oblivious asynchronous message-passing systems, where an adversary may arbitrarily corrupt message contents.

Frei~et~al.~(DISC 2024) presented a uniform terminating leader election algorithm for oriented rings in this setting, with message complexity $O(n \cdot  \myID_{\max})$ on a ring of size $n$, where $\myID_{\max}$ is the largest identifier in the system, this result has been recently extended by Chalopin~et~al.~(DISC 2025) to unoriented rings.  

In this paper, we investigate the message complexity of leader election on ring networks in the content-oblivious model, showing that no uniform algorithm can solve the problem if each process is limited to sending a constant number of messages in one direction.

Interestingly, this limitation hinges on the \emph{uniformity}
assumption. In the \emph{non-uniform} setting, where processes know an
upper bound $U \geq n$ on the ring size, we present an algorithm with
message complexity $O(n \cdot U \cdot  \myID_{\min})$, in which each
process sends $O(U \cdot \myID_{\min})$ messages clockwise and only three
messages counter-clockwise. Here, $\myID_{\min}$ is the smallest
identifier in the system. This dependence on the identifiers compares favorably
with the dependence on $\myID_{\max}$ of Frei~et~al.

We also show a non-uniform algorithm where each process
sends $O(U \cdot \log\myID_{\min})$ messages in one direction and
$O(\log\myID_{\min})$ in the other. The factor $\log \myID_{\min}$ is \emph{optimal}, matching the lower bound of
Frei~et~al.

Finally, in the anonymous setting, where processes do not have identifiers, we propose a randomized algorithm where each process sends only $O(\log^2 U)$ messages, with a success probability of $1 - U^{-c}$.

\end{abstract}

\thispagestyle{empty}
\newpage
\thispagestyle{empty}
\tableofcontents
\newpage
\pagenumbering{arabic}

\section{Introduction}\label{sec1}

The field of distributed computing is characterized by the study of a wide range of failure models, from failures affecting specific processes (such as crash-stop failures and Byzantine failures \cite{cachin2011introduction,raynal2018fault}) to failures affecting communication channels (such as message omission failures, message addition failures or message corruptions \cite{santoro1989time}). 
The recent \cite{censor2023distributed} introduced an interesting failure model for channels: the fully-defective network. In this model, processes do not fail, but all messages in transit may be arbitrarily corrupted by an adversary. While the content of messages may be altered, messages cannot be created or destroyed.

This model is therefore content-oblivious, meaning that messages carry no reliable content, not even the identity of the sender. The only information a message conveys is its existence. 

Surprisingly, \cite{censor2023distributed} showed that if a leader process is initially known and when the topology is two-edge-connected, it is possible to simulate a reliable, uncorrupted message-passing network, even in asynchronous settings, within this fully-defective model. Interestingly, they have also shown that if the network is not two-edge-connected, then non-trivial computations cannot be carried out in the content-oblivious model. 

This result makes studying leader election algorithms for this model particularly compelling. As a matter of fact, if a leader can be elected in a certain topology, with an algorithm that is terminating and where it is possible to distinguish in some way the messages of the leader election algorithm, then on such topology any asynchronous algorithm can be simulated even if all messages are arbitrarily corrupted.  
In this regard \cite{censor2023distributed}
conjectured that having a leader was a necessary requirement and that, on general networks, a leader election algorithm was impossible to build in a content-oblivious network. 

Frei et al.~\cite{frei2024content} disproved the aforementioned conjecture for a notable family of two-edge-connected networks: oriented ring topologies. 
In an oriented ring, processes have a common notion of clockwise and counter-clockwise direction, that is, each edge of a process is locally associated to a port number that is either $0$ or $1$, and such labeling is globally consistent (more details in the System Model Section). 
More specifically,~\cite{frei2024content} proposed a terminating leader election algorithm with a message complexity of $O(n \cdot \myID_{\max})$, where $\myID_{\max}$ is the maximum identifier in the system.  The algorithm leads to \cite[Corollary 4]{frei2024content}, which shows that the presence of unique identifiers is enough to simulate any asynchronous algorithm on a content-oblivious oriented ring. 

This result has been recently extended in \cite{cha2025content} to unoriented rings, that are rings where there is no agreement on clockwise (CW) or counter-clockwise (CCW) direction, showing that, when unique IDs are present, content-oblivious rings are computationally equivalent to non-faulty rings. 

 A natural question is whether the assumption of distinct node $\id$s can be removed. In such an anonymous setting, nodes must rely on randomness to break symmetry. The classic impossibility result of~\cite{Itai90negative} shows that no uniform algorithm can measure the size of a ring when nodes are anonymous. In our setting, this implies that no uniform algorithm can both elect a leader and guarantee explicit termination on anonymous rings, since once a leader is elected, the ring size can be learned using the general algorithm simulation from the work~\cite{censor2023distributed}.

For such a reason \cite{frei2024content} provides a randomized leader election algorithm that uses $n^{O(c^2)}$ messages and succeeds with probability $1 - n^{-c}$, but does not explicitly terminate.

\subsection{Contributions}

Previous works~\cite{frei2024content,cha2025content} focused on feasibility results, and all existing algorithms have a message complexity that grows linearly with the values of process identifiers. Our focus is to study the message complexity of leader election in oriented rings, aiming to design algorithms with improved performance and where the dependence on the identifiers is not linear.

 Regarding necessary complexity, in Section~\ref{sec:imp} we show that a uniform algorithm, even in oriented rings, cannot elect a leader while sending only a constant number of messages in a single direction. More precisely:

\begin{restatable}{theorem}{boundth}
For any uniform leader election algorithm and any bound $b \in \mathbb{N}$ there exists an oriented ring where at least one process sends more than $b$ messages in each direction. 
\label{thm:boundth}
\end{restatable}

Interestingly, we show that non-uniform algorithms, where processes know an upper bound $U$ on the network size $n$, can circumvent this impossibility result. In particular, we present in Section~\ref{sec:cdir} an algorithm that sends a constant number of messages in one direction, allowing us to state:

\begin{restatable}{theorem}{algone}
There exists a quiescently terminating  non-uniform leader election algorithm for  oriented rings in which each process sends $O(U \cdot \myID_{\min})$ messages clockwise and three messages counter-clockwise, for a total of $O(n \cdot U \cdot \myID_{\min})$ messages.
\label{thm:alg1}
\end{restatable}

Here $\myID_{min}$ is the minimum identifier. We then present in Section~\ref{sec:logalg} a non-uniform algorithm for oriented rings that has a logarithmic dependency on identifier values: 

\begin{restatable}{theorem}{algtwo}
There exists a quiescently terminating non-uniform leader election algorithm for  oriented rings in which each process sends $O(U \cdot \log \myID_{\min})$ messages clockwise and $O(\log \myID_{\min})$ messages counter-clockwise, for a total of $O(n \cdot U \cdot \log \myID_{\min})$ messages.
\label{thm:alg2}
\end{restatable}

We stress that the dependency on $\log(\myID_{\min})$ is optimal. This follows from~\cite{frei2024content}, which shows that a logarithmic dependency on the identifier values is unavoidable for any uniform leader election algorithm in the content-oblivious setting. In Appendix~\ref{sect: lower bound} we show that such lower bound holds also for non-uniform algorithms when the exact network size is known.

The key technical novelty in our non-uniform algorithms is the design of a global synchronization mechanism that leverages the known upper bound on the network size to keep processes partially synchronized across algorithmic phases. This mechanism enables processes to compare their identifiers bit by bit, leading to an algorithm whose complexity depends only on the size of the identifiers (and not on their values).

All the above algorithms require unique identifiers. When identifiers are absent, the non-uniformity assumption still allows us to bypass the impossibility result of~\cite{Itai90negative}. In this case, we present a terminating randomized algorithm in Section~\ref{sect: rand}:

\begin{restatable}{theorem}{mainthmrand}\label{thm:randmaintheorem}
  For any constant $c > 0$, there is a {randomized} 
  leader election algorithm with message complexity
  $O(n\cdot \log^2 U)$ in an anonymous oriented ring $C$ that succeeds
  with probability at least $1-U^{-c}$, where $U \geq n$ is a known
  upper bound on the number of nodes $n$.
\end{restatable}

We emphasize that this randomized algorithm, presented in Section~\ref{sect: rand}, improves upon the one in~\cite{frei2024content}, whose complexity includes an exponential dependence on $c$. In contrast, our randomized algorithm has only a quadratic dependence on $c$, which is hidden in the asymptotic notation.

We also stress that the algorithm of~\cite{frei2024content} is non-terminating, as required by the impossibility results of~\cite{Itai90negative}. In contrast, our algorithm terminates since it operates in the non-uniform setting. A detailed comparison between the two algorithms is provided in Appendix~\ref{sect:adaptedalgo}, where we also explain why a direct adaptation of the approach in~\cite{frei2024content} to the non-uniform setting yields termination but does not remove the exponential dependency on $c$.

A comparison of our algorithm and the previous literature can be found in Table~\ref{tab:results}.


\begin{table}[ht]
\centering
\begin{tabular}{|c|c|c|ll|l|}
\hline
\multirow{2}{*}{\textbf{IDs}} & \multirow{2}{*}{\textbf{Deterministic}} & \multirow{2}{*}{\textbf{Uniform}} & \multicolumn{2}{c|}{\textbf{Messages per process}}                                    & \multicolumn{1}{c|}{\multirow{2}{*}{\textbf{Reference}}} \\ \cline{4-5}
                              &                                         &                                   & \multicolumn{1}{c|}{\textbf{CW}}                    & \multicolumn{1}{c|}{\textbf{CCW}} & \multicolumn{1}{c|}{}                                    \\ \hline
\checkmark                    & \checkmark                              & \checkmark                        & \multicolumn{1}{l|}{$O(\myID_{\max})$}              & $O(\myID_{\max})$                 & \cite{frei2024content}                                   \\ \hline
\checkmark                    & \checkmark                              & -                                 & \multicolumn{1}{l|}{$O(U \cdot \myID_{\min})$}      & $3$                               & Section~\ref{sec:cdir}                                      \\ \hline
\checkmark                    & \checkmark                              & -                                 & \multicolumn{1}{l|}{$O(U \cdot\log(\myID_{\min}))$} & $O(\log(\myID_{\min}))$           & Section~\ref{sec:logalg}                                    \\ \hline
-                             & -                                       & -                                 & \multicolumn{1}{l|}{$O(\log^2(U))$}                 & $O(\log(U))$                      & Section~\ref{sect: rand}                                    \\ \hline
\end{tabular}
\caption{Comparison of leader election algorithms for content-oblivious oriented rings; $n$ is the size of the ring.}
\label{tab:results}
\end{table}

\section{Related Work}

Electing a leader in the classical message-passing setting has been extensively studied \cite{attiya2004distributed,flocchini2004sorting,kutten2020singularly,santoro2007design}; a seminal early result is the necessity of unique identifiers in anonymous symmetric networks, such as rings~\cite{angluin1980local}. In the following we will focus on the works that study leader election in the settings that are more similar to us, that are ring networks of correct processes with unique identifiers. The reader that is interested in leader election with faulty processes or in anonymous networks can refer to \cite{attiya2004distributed,BVelection,boldi2002fibrations, cachin2011introduction, chalopin2012election, raynal2018fault, santoro2007design,YKsolvable}.

When identifiers are present and the ring is asynchronous, leader election requires $\Omega(n  \cdot \log n)$ messages; this bound holds regardless of the knowledge of $n$ or the orientation of the ring~\cite{burns1980formal,korach1984lower}. This bound is tight, as demonstrated by the unoriented, asynchronous, and uniform algorithm of~\cite{hirschberg1980decentralized}. When the system is synchronous, the best achievable message complexity is $O(n)$~\cite{elruby1991linear,frederickson1987electing}.

Computability in content-oblivious networks is still in its infancy, with only three papers published so far \cite{censor2023distributed, frei2024content,cha2025content}. In \cite{censor2023distributed}, an algorithm is presented that simulates any asynchronous algorithm for reliable message-passing in two-edge-connected networks. The two-edge connectivity assumption is crucial, as the paper also shows that two processes connected by a single link cannot compute any non-trivial predicate, including leader election.

In \cite{frei2024content}, apart from a terminating leader election algorithm for oriented rings, they show that in the uniform setting for any algorithm there exists an assignment of identifiers to processes such that the  algorithm has to send $\Omega(n \cdot \log(\frac{\myID_{\max}}{n}))$ messages. Recall that in Appendix \ref{sect: lower bound} we adapt such result for the uniform setting. 

An interesting corollary of this bound is that any content-oblivious leader election algorithm must have a complexity explicitly dependent on the identifiers assigned to processes, a dependence not present in the standard message-passing model.
Moreover, this result shows that the logarithmic factor in the complexity $O(n \cdot U \cdot  \log(\myID_{\min}))$ of our non-uniform algorithm is necessary.\footnote{Note that our dependency on $\myID_{\min}$ is not in contradiction with the presence of $\myID_{\max}$ in the lower bound of \cite{frei2024content}. Please refer to Appendix \ref{sect: lower bound} for more details.}

Our impossibility result complements this lower bound by showing that any uniform leader election algorithm must not only send $\Omega(n  \cdot \log(\frac{\myID_{\max}}{n}))$ messages, but also cannot send just a constant number of messages in one direction.

Finally, the recent \cite{cha2025content} has shown a leader election algorithm for non-uniform general two-edge-connected networks with a message complexity of $O(m \cdot U \cdot \idmin)$, where $m$ is the number of edges in the networks, and a uniform leader election algorithm for oriented rings with a message complexity of $O(n \cdot \idmax)$.  

A related line of research concerns the beeping model~\cite{cornejo2010deploying, CASTEIGTS201920, CASTEIGTS201932,VacusZiccardi2025}, where in each round a process may either emit a beep to its neighbors or listen. This model is inherently synchronous. We emphasize that synchrony enables the use of time as a means of communication, even without explicit messages (see~\cite[Chapter 6]{santoro2007design}), which is not possible in our model.
Therefore, the techniques developed for leader election in the beeping model \cite{dufoulon2018,forster2014deterministic,czumaj2019leader} cannot be applied to our setting.

\section{System Model}\label{sec:model}
We consider a distributed system composed of a set of $n$ processes $P: \{p_0, p_1, \ldots, p_{n-1}\}$ that communicate on a ring network by sending messages to each other. More precisely, a process $p$ has two local communication ports: port $0$ and port $1$. By means of a specific port, a process is able to send messages to and receive messages from one of its neighbors.

\subparagraph{IDs and Anonymity.}
Our deterministic algorithms consider processes that have unique identifiers that are arbitrarily selected from $\N$. In this case, for each $j \in [0,n-1]$, let $\myID_j$ be
the identifier of process $p_j$. 
Our randomized algorithm works on {\em anonymous systems}, where processes do not have identifiers and all start in the same state.

\subparagraph{Asynchronous System.}
We assume an {\em asynchronous system}; that is, each message has an unpredictable delay that is finite but unbounded. A message cannot be lost. Messages are buffered, in the sense that if several messages are delivered to a process at the same time, they are stored in a local buffer and can then be retrieved by the destination process at its discretion. The local computation time at each process may be arbitrary but bounded, and processes do not share a notion of common time. In the analysis of our algorithms we assume that a process performs all its non-blocking action instantaneously (i.e., all actions different from receiving a message); this is not restrictive as the local computation delay can be conflated with the message delays.

\subparagraph{Content-Oblivious Algorithms.}
Since we are interested in {\em content-oblivious algorithms}, we assume that messages do not carry any information apart from their existence. That is, we can imagine that each message is an empty string, or that it is simply a {\em pulse}~\cite{frei2024content}.

We adopt this model to represent an {\em adversary} that is able to corrupt the content of all messages traveling through the network (but we stress that such an adversary is not able to erase messages from the network or inject additional messages). Processes are correct and do not experience any kind of failure.

\subparagraph{Oriented  Rings.}
In this paper, we consider {\em oriented} rings. A ring is oriented if processes share a common notion of {\em clockwise} (CW) and {\em counter-clockwise} (CCW) orientation. More precisely, we consider
a ring $(p_0,p_1, \ldots, p_{n-1})$ and assume that for each
$j \in [0,n-1]$, at process $p_j$, port $0$ leads to $p_{j+1}$  and port $1$ leads to $p_{j-1}$ (where
addition is made modulo $n$). We will say that a message is traveling in the clockwise direction if it is sent on port $0$ and received on port $1$ (i.e., a message goes from $p_i$ to $p_{i+1}$); conversely, a message travels in the counter-clockwise direction if it is sent on port $1$ and received on port $0$.

\subparagraph{Uniform and Non-Uniform Algorithms.}
An algorithm is {\em uniform} if it works for all possible network sizes; i.e., processes do not know the total number of processes $n$ or a bound on it. An algorithm is {\em non-uniform} if processes know the value of $n$. In this paper, we consider as non-uniform a slightly weaker assumption, where processes know an upper bound $U \geq n$.

\subparagraph{Message Complexity.} We measure the performance of our algorithm by analyzing their {\em message complexity}. That is the number of all messages sent. For oriented algorithms, we further discriminate between the numbers of clockwise messages and counter-clockwise messages.  In Section \ref{sec:imp} we will show that it is not possible to devise a uniform algorithm for oriented rings that sends a constant number of messages in one direction.

\subparagraph{Leader Election Problem.} In this problem, processes must elect one them as the {\em leader} and all others are designated as {\em non-leader} processes. We consider {\em terminating leader election}, in the sense that every process must eventually enter either a \leader or \nonleader state, and both of these states are final—meaning no further state changes are possible, and the process terminates the execution of the algorithm.

Moreover, we require that once the last process enters its final state, there are no messages in transit in the network. This is known as {\em quiescent termination}, as defined in~\cite{frei2024content}. In our algorithms, the last process to enter its final state is always the unique \leader.

This property enables easy composability of algorithms: once the leader enters the final state, it knows that the network is ``at rest'' and can begin sending messages for a subsequent algorithm (see \cite[Section 1.1]{frei2024content}).

\section{An Algorithm That Sends \texorpdfstring{$O(n \cdot U \cdot \log(\myID_{\min}))$}{O(n U log IDmin)}
  Messages}\label{sec:logalg}

In this section, we describe a non-uniform algorithm in which each process sends
$O(U \cdot \log \myID_{\min})$ messages. The algorithm pseudocode is in Algorithm \ref{algo:orientedknown}.

The algorithm elects the process with the minimum identifier. We first
transform the identifiers of the processes to ensure they satisfy some
properties (see details below). Then our algorithm proceeds in
elimination rounds: in each round $i$, the set of active processes
compares the $i$-th bit of their (encoded) identifiers, eliminating
all processes that do not have the minimum bit at that position.

The bound $U$ on the network size is used to distinguish different
rounds of the algorithm. In fact, our algorithm takes a parameter
$\myUB$ and assumes that at each round $i$, either all active
processes have the same bit, or there are at most $\myUB-1$ active
processes appearing consecutively (ignoring the inactive processes)
that have the same $i$-th bit.  Observe that if we set $\myUB = U$,
then this property is always true.

At each round, each process sends and receives messages in the
clockwise direction for $\myUB$ steps (see Synchronization Phase at
Lines 4--8, in blue). Then, the processes with bit $0$ send and
receive one message in the counter-clockwise direction, blocking
messages in the clockwise direction (see Zero Signaling Phase at Lines
10--11, in purple). Meanwhile, the processes with bit $1$ attempt to
continue sending and receiving messages in the clockwise direction for
another $\myUB$ steps. However, they become inactive if they receive a
counter-clockwise message, indicating the presence of a process with
bit $0$ (see No-Zero Checking Phase at Lines 17--22, in orange). Note
that if not all processes have bit $1$, then by our assumption, each
process $p$ with bit $1$ can receive (and relay) at most $\myUB-1$
clockwise messages in this phase as there is a process $p'$ with bit
$0$ at distance at most $\myUB-1$ from $p$ in the counter-clockwise
direction and this process $p'$ is blocking all clockwise messages.

An inactive process will relay messages, and while doing so, it
eliminates the first clockwise message it receives. This ensures that,
once all inactive processes have eliminated a message, the total
number of messages in the network equals the number of active
processes (see Inactive Phase at Lines 24--25, in red as well as the
function \relay described in Algorithm~\ref{alg:relay}).

The algorithm terminates after $\len(\myID_{\min}) = O(\log \myID_{\min})$ rounds. At this point, the only remaining active process is the one with the minimum identifier. This process can detect this locally and communicates it to all other processes (which are, by construction, inactive) by sending an additional message in the counter-clockwise direction (see Termination Phase at Lines 13--15, in green). 

All other processes detect termination when they receive $2$ consecutive messages in the counter-clockwise direction: one of these messages was sent during the Zero Signaling Phase by the process with the minimum identifier, the additional message is sent during the Termination Phase.

\begin{algorithm}
\footnotesize
  \DontPrintSemicolon
  \SetKw{Send}{send}
  \SetKw{Continue}{continue}
  \SetKw{Break}{break}
  \SetKw{Receive}{receive}
  \SetKw{UnorientedRelayAndWaitTermination}{UnorientedRelayAndWaitTermination}
  \SetKwData{Status}{status}
  \SetKwFunction{Undecided}{Undecided}
  \SetKwFunction{Leader}{Leader}
  \SetKwFunction{NonLeader}{Non-Leader}
  \SetKwFunction{Relay}{KillACWMessageAndRelay}
  \SetKwFunction{Sync}{Synchronize}
  \SetKwData{Count}{received-CW}
  \SetKwData{Num}{num}
  \SetKwData{True}{true}
  \SetKwData{ID}{ID}
  \SetKwData{SizeNetwork}{d}
  \SetKwData{Bit}{bit}
  \SetKwData{Len}{length}
  \SetKw{Wait}{Wait until}

  $\Count \leftarrow 0$\;
  
  \For{$i \leftarrow 1$ \KwTo $\Len(\ID)$}
  {
    $b \leftarrow \Bit(\ID, i)$\;
    \tcc{Synchronization: All active processes send a message on port $0$ and relay at
      least $\SizeNetwork -1$ messages moving in this direction} 
    \color{blue}
    \Repeat{$\Count = (2i-1) \cdot \SizeNetwork$}
    {
      \Send a message on port $0$\;
      \Receive a message on port $1$\;
      $\Count \leftarrow \Count + 1$\;
    }
    \color{black}
    \eIf{$b=0$}
    {      
      \tcc{Zero Signaling: The bit at position $i$ in $\ID$ is $0$ and
        the process will stay active. To inform all active
        processes that have a bit $1$ at position $i$ in their $\ID$
        that they are eliminated, the process sends a message on
        port $1$ and waits for a message on port $0$.}  
      \color{violet}
        \Send a message on port $1$\;
        \Receive a message on port $0$\;
      \color{black}
      \If{$i = \Len(\ID)$}
      {
        \tcc{Termination: The process has won all the rounds and is the only one
          that is still active}
        \color{teal}
        \Send a message on port $1$\;
        \Receive a message on port $0$\;
        \Return{\Leader}
        \color{black}        

      }

    }
    {
      \tcc{No-Zero Checking: The process has a bit $1$ at position $i$ in its $\ID$ and
        it will advance to the next round. We check that this is the
        case for all active processes using direction $0$ or we are
        eliminated by a message moving in the other direction.} 
      \color{orange}
      \Repeat{$q=0$ or $(\Count = 2i \cdot \SizeNetwork)$}
      {
        \Send a message on port $0$\;
        \Receive a message on port $q$\;
        \If{$q=1$}{
          $\Count \leftarrow \Count +1$\;
        }
      }
      \color{black}
      \If{$q=0$}{
        \tcc{Inactive: There is another process with a bit  $1$ at position $i$
          in its $\ID$. The process forwards this message and becomes
          a relaying process. While relaying, it kills the first
          message received on port $1$ to cancel the message sent on
          port $0$ at Line 22.}
        \color{red}
        \Send a message on port $1$\;
        \Relay()\;
        \Return{\NonLeader}\;
        \color{black}
      }

    }
  }

  \caption{LogarithmicLeaderElection(\myID,\myUB)} 
  \label{algo:orientedknown}
\end{algorithm}

\begin{algorithm}
\footnotesize
  \DontPrintSemicolon
  \SetKw{Send}{send}
  \SetKw{Receive}{receive}  
  \SetKw{Continue}{continue}
  \SetKwData{CountZero}{consecutive-CCW-msgs}
  \SetKwData{Threshold}{threshold}
  \SetKwData{SizeNetwork}{size}
  \SetKwData{Killed}{killed}
  $\CountZero \leftarrow 0$\;
  $\Killed \leftarrow 0$ \;
  \Repeat
  {$|\CountZero| = 2$} 
  {
    \Receive a message on port $q$\;
    \eIf{$q = 1$}
    {
    \If{$\Killed=0$}{
        $\Killed \leftarrow 1$ \;
        \Continue 
    }
    $\CountZero \leftarrow 0$}
    {$\CountZero \leftarrow \CountZero +1$}
    \Send a message on port $1-q$\;
  }
  \caption{{\relay}}\label{alg:relay}
\end{algorithm}

\subsection{Correctness} 

We say that a process that has sent more messages on port $0$ than it
has received on port $1$ is \emph{CW-unbalanced}. Our algorithms are designed in such a way that if $p$
is CW-unbalanced, then it has sent precisely one more message on port $0$
than it has received on port $1$. A process that is not CW-unbalanced
is said to be \emph{CW-balanced}.

Our algorithm assumes that identifiers satisfy certain properties. For each identifier $\myID$, let $\len(\myID)$ be its length, and let $\ell_{\min}$ denote the length of the shortest identifier $\myID_{\min}$.

Given an identifier $\myID$ of length $\ell$ and an integer
$1 \leq i \leq \ell$, let $\pref_i(\myID)$ be the number obtained when
considering only the first bits $i$  of $\myID$. Assuming that the
set of identifiers of the processes on the ring is $S$ and given an
integer $i \leq \ell_{\min}$ , we let $\actif_i$ be the set of
processes that have identifiers in
$\{\myID \in S : \forall \myID'\in S, \pref_i(\myID) \leq
\pref_i(\myID')\}$. For each $i$, we partition $\actif_i$ in two sets:
$\actif_i^0$ (respectively, $\actif_i^1$) is the set of processes in
$\actif_i$ such that the $i$-th bit of their identifier is $0$
(respectively, $1$). Observe that for each
$0 \leq i \leq \ell_{\min}-1$, $\actif_{i+1} = \actif_i^0$ if
$\actif_i^0 \neq \emptyset$ and $\actif_{i+1} = \actif_i^1 = \actif_i$
otherwise.

For each process $p$ in $\actif_i$, we define $\suc_i(p,k)$
(respectively, $\pred_i(p,k)$) to be the $k$-th process
$p' \in \actif_i$ encountered when moving clockwise (respectively,
counter-clockwise) along the ring from $p$. We denote $\suc_i(p,1)$ by
$\suc_i(p)$ and we call it the successor of $p$ at round
$i$. Similarly, we let $\pred_i(p) = \pred_i(p,1)$ and call it the
predecessor of $p$ at round $i$. Note that if there are $n_i$
processes in $\actif_i$, we let
$\pred_{i}(p,k) = \pred_i(p, k \mod n_i)$ and
$\suc_{i}(p,k) = \suc_i(p, k \mod n_i)$ when $k \geq n_i$. In
particular, if $\actif_i = \{p\}$, then $\pred_i(p) = \suc_i(p) = p$.

We say that the distribution of identifiers on the ring is
\emph{$\myUB$-scattered} if for each $i \leq \ell_{\min}$, either
$\actif_i^0$ is empty or for each process $p \in \actif_i^1$, there
exists $1 \leq k,k' < \myUB$ such that
$\pred_i(p,k), \suc_i(p,k') \in \actif_i^0$. Note that in any ring of
size at most $U$, any distribution of identifiers is $U$-scattered.

We say that the distribution of identifiers on the ring is
\emph{$0$-ended} if the last bit of each identifier is $0$.  We say
that the distribution of identifiers on the ring is \emph{strongly
  prefix-free} if for any two identifiers $\myID_1 \leq \myID_2$ of
respective lengths $\ell_1 \leq \ell_2$, we have
$\myID_1 \leq  \pref_{\ell_1}(\myID_2)$. Observe that starting from
an arbitrary distribution of identifiers, one can obtain a $0$-ended
strongly prefix-free distribution of identifiers as follows: we encode
each identifier $\myID$ of length $\ell = \len(\myID)$ as
$1^{\ell}0\cdot\myID\cdot 0$.

In the following, we assume that the distribution of identifiers on
the ring is $\myUB$-scattered, $0$-ended, and strongly prefix-free.

We say that a process is \emph{active} if it is not executing Lines
24-26, and that it is \emph{inactive} otherwise. We say that a process
is active at round $i$ if it enters the loop at Lines 4--8 in the
$i$-th iteration of the main loop of the algorithm.

Note that an active process can be CW-unbalanced (if it is executing
the loop at Lines 4--8 or the loop at Lines 17--22) or CW-balanced
(otherwise). An inactive process is initially CW-unbalanced until it
kills a message at Line 4 of Function \relay when $q=1$ where it
becomes CW-balanced. Observe that a CW-balanced process has sent the
same number of messages on port $0$ as it has
received on port $1$ (recall that we assume that when a process is
activated, it performs all the local computations it can).

We say that a message is in \emph{transit} between $p_i$ and $p_j$ if
it has been sent by some process $p_\ell$ and not yet received by
$p_{\ell+1}$ where $p_\ell$ and $p_{\ell+1}$ are consecutive processes
appearing between $p_i$ and $p_j$ when moving clockwise.  Notice that
the number of messages in transit in the network is precisely the
number of CW-unbalanced processes.

Since each message sent is either received or in transit, the following property holds:  

\begin{observation}\label{obs:diff-recv}
  Given two processes $p_i$ and $p_j$, let $k_i$ be the number of
  clockwise messages received by $p_i$ and $k_j$ be the number of
  clockwise messages received by $p_j$. Then the number of clockwise messages in
  transit between $p_i$ and $p_j$ is exactly $k_i - k_j + u_{ij}$
  where $u_{ij}$ is the number of CW-unbalanced processes between
  $p_i$ and $p_{j-1}$ (when moving clockwise), including $p_i$ and
  $p_{j-1}$.
\end{observation}

\begin{lemma}\label{lemma:forwardbound}
  Suppose that each process initially sends a clockwise message (that is a message on port $0$) and
  that each process forwards all clockwise messages it receives  (each time it
  receives a message on port $1$ it sends a message on port $0$).  If a process has received $k$ clockwise
  messages, and $k \geq n$ then each other process has sent at least
  $k-n+1$ clockwise messages and received at least $k-n$ clockwise messages.
\end{lemma}

\begin{proof}
  For the purpose of the proof we identify each message on the network
  with the index of its initial sender, these indices are just used
  for proof purposes, so we have $n$ clockwise messages 
  $m_{0},m_1,\ldots,m_{n-1}$ that are forwarded on the ring.  

  Without loss of generality, assume that $p_0$ has received $k$
  clockwise messages. For any $i \in [0,k]$, let $t_i$ be the number of times
  $p_0$ has received message $m_i$ and observe that
  $k = \sum_{i \in [0,n-1]}t_i$.  Notice that if $p_0$ has received
  $t_i$ times the message $m_{i}$ that was initially sent by $p_i$,
  then each process $p_j$ has sent the message $m_i$ at least $t_i$
  times if $i \leq j$ and at least $t_i-1$ times otherwise.

  Consequently, each process different from $p_0$ has sent at least
  $ t_0+\sum_{i \in [1,n-1]}(t_i-1) =(\sum_{i \in [0,n-1]}t_i) -n + 1
  = k-n+1$ clockwise messages.  Since for each process, the difference between
  the numbers of clockwise messages sent and of clockwise messages received is at most
  $1$, each process has then received at least $k-n$ clockwise messages.
\end{proof}

\begin{lemma}\label{lemma:forwardbound2} 
  Suppose that a process $p$ has received $k \geq n$ clockwise
  messages. Assume that no process terminates the algorithm before it
  has received $k$ clockwise messages, then:
  \begin{enumerate}[label=(\arabic*)]
  \item\label{lfb:1} $p$ has sent $k$ or $k+1$ clockwise messages,
  \item\label{lfb:2} Each process has received at least $k-n$ and sent
    at least $k-n+1$ clockwise messages,
  \item\label{lfb:3} Each other process eventually receives $k$
    clockwise messages.
  \end{enumerate}
\end{lemma}

\begin{proof}
  Observe that when an active process $p$ executes the algorithm, each
  send on port $0$ (of a clockwise message) is followed by a receive
  on port $1$ (of a clockwise message) except at Line 19 if $q=0$.

  However, note that in this case, $p$ will kill the next clockwise
  message it will receive (at line 4 of Function \relay). And then, as
  an inactive process, $p$ will only relay clockwise messages, i.e.,
  it will stay CW-balanced all along the rest of the execution. The
  number of messages circulating in the clockwise direction is thus
  bounded by $n$ and (\ref{lfb:1}) holds.  Analyzing only messages
  traveling in the clockwise direction, by
  Lemma~\ref{lemma:forwardbound}, (\ref{lfb:2}) also holds.

  Suppose without loss of generality that $p_0$ has received $k$
  clockwise messages and assume that there is a process $p_j$ that
  never receives $k$ clockwise messages. By choosing $j$ to be
  minimum, we can assume that $p_{j-1}$ eventually receives $k$
  clockwise messages. At this point, by (\ref{lfb:1}), $p_{j-1}$ has
  sent at least $k$ clockwise messages to $p_j$. Consequently, if
  $p_j$ has received $k' < k$ clockwise messages, by~(\ref{lfb:1}),
  there are at least $k-k'$ messages in transit between $p_{j-1}$ and
  $p_j$ that will eventually be delivered to $p_j$. This
  establishes~(\ref{lfb:3}).
\end{proof}
\color{black}

In order to prove the correctness of our algorithm, we establish some
invariants satisfied along the execution.

\begin{proposition}\label{prop-IH}
  Consider a ring with distribution of identifiers that is
  $\myUB$-scattered, $0$-ended and strongly prefix-free. Let
  $p_{0}$ be the process with the minimum identifier
  $\myID_{\min}$. Then for each
  $1 \leq i \leq \ell_{\min} = \len(\myID_{\min})$, the following
  holds:
  \begin{enumerate}[label=(\arabic*)]
  \item\label{IH-1} a process $p$ executes  the loop at Lines 4--8 at
    round $i$ if and only if $p \in \actif_i$;    
  \item\label{IH-2} When a process $p$ exits the loop at Lines 4--8 at
    round $i$, then the processes between $\pred_{i}(p)$ and $p$
    (excluding $\pred_{i}(p)$ and $p$) are inactive and CW-balanced. 
  \item\label{IH-2bis} When a process $p$ exits the loop at Lines 4--8
    at round $i \geq 2$, $\pred_{i}(p)$ has exited the loop at Lines
    4--8 at round $i-1$;
  \item\label{IH-3} every process $p \in \actif_i^0$ executes Lines 10 and 11;
  \item\label{IH-3bis} every counter-clockwise message created by a
    process $p$ at Line 10 at round $i$ is killed by $\pred_{i}(p)$
    executing Line 11 at round $i$; 
  \item\label{IH-4} if $\actif_i^0 \neq \emptyset$, every process
    $p \in \actif_i^1$ exits the loop at Lines 17--22 with $q = 0$;
  \item\label{IH-4bis} if $\actif_i^0 = \emptyset$, every process
    $p \in \actif_i^1$ exits the loop at Lines 17--22 with $q = 1$ and
    after having received $2i\myUB$ clockwise messages;
  \item\label{IH-5} the only process that terminates while being
    active (i.e., at Line 15) is $p_{0}$ with identifier
    $\myID_{\min}$; when this happens all other processes have
    terminated and there is no message in transit;
  \item\label{IH-6} no inactive process terminates (at Line 26) before
    the process $p_{0}$ executed Line 13.
  \end{enumerate}
\end{proposition}

We prove the proposition by induction on $i$ in a series of lemmas.

\begin{lemma}[\ref{IH-1}]
  The set $\actif_i$ is precisely the set of processes active at round
  $i$.
\end{lemma}

\begin{proof}
  Observe that initially, all processes are in $\actif_1$ and they are
  all active at round $1$. Assume now that the induction hypothesis
  holds for round $i$ and consider round $i+1$. Note first that only
  processes active at round $i$ can remain active at round
  $i+1$. Then, if $\actif_{i}^0 \neq \emptyset$, by (\ref{IH-3}), all
  processes in $\actif_{i}^0$ are active at round $i+1$ and by
  (\ref{IH-4}), all processes in $\actif_{i}^1$ become inactive at
  round $i$. Now, if $\actif_{i}^0 = \emptyset$, by (\ref{IH-4}), all
  processes in $\actif_i = \actif_i^1$ remain active at round
  $i+1$. This shows that the set of active processes at round $i+1$ is
  precisely $A_{i+1}$.
\end{proof}

\begin{lemma}[\ref{IH-2},\ref{IH-2bis}]
  Consider $p,p' \in \actif_i$ such that $p' = \pred_i(p)$. When $p$
  exits the loop at Lines 4--8 at round $i$, there are no
  CW-unbalanced processes between $p'$ and $p$ (excluding
  $p'$). Moreover, if $i \geq 2$, then $p'$ has exited the loop at
  Lines 4--8 at round $i-1$.
\end{lemma}

\begin{proof}
  Observe that if $i = 1$, then there are no processes between
  $\pred_1(p)$ and $p$ and there is nothing to prove. Assume now that
  the induction hypothesis holds for round $i$ and consider round
  $i+1$. Let $p,p' \in \actif_{i+1}$ such that $p' = \pred_{i+1}(p)$,
  and let $t$ be the moment at which $p \in A_{i+1}$ exits the loop at
  Lines 4--8 at round $i+1$.

  Let $\lost$ be the set of processes in $\actif_i$ between $p'$ and
  $p$ (excluding $p'$ and $p$). Note that if
  $\actif_{i+1} = \actif_{i}$, then $\lost = \emptyset$ and that if
  $\actif_{i+1} = \actif_i^0$, then $\lost \leq \myUB-1$ as the
  distribution of identifiers is $\myUB$-scattered.  We claim that at
  time $t$, all processes in $\lost \cup \{p'\}$ have exited the loop
  at Lines 4--8 at round $i$. Suppose it is not the case and consider
  the process $q \in \lost$ that is the closest to $p$ such that at
  time $t$, it has not exited the loop at Lines 4--8 at round $i$.  By
  induction hypothesis (\ref{IH-2}), all processes
  between $q$ and $p$ that are not in $\actif_i$ are CW-balanced at
  time $t$ and thus there are at most $\myUB-1$ processes between $q$
  and $p$ that are CW-unbalanced at time $t$. Since $p$ has received
  $(2(i+1)-1)\myUB = (2i+1)\myUB$ clockwise messages at time $t$, by
  Observation~\ref{obs:diff-recv}, $q$ has received at least
  $(2i+1)\myUB - \myUB = 2i \myUB > (2i-1)\myUB$ clockwise
  messages. Consequently, $q$ has exited the loop at Lines 4--8 at
  round $i$, contradicting our choice of $q$. This establishes that at
  time $t$, all processes in $\lost \cup \{p'\}$ have exited the loop
  at Lines 4--8 at round $i$, establishing~(\ref{IH-2bis}) for round
  $i+1$.

  Note that if $\lost = \emptyset$, then when $p$ exits the loop at
  round $i$, all processes between $p'$ and $p$ are inactive and
  CW-balanced. Since inactive processes that are CW-balanced remain
  inactive and CW-balanced, property (\ref{IH-2}) also holds at round
  $i+1$ in this case.  Suppose now that
  $p, p' \in \actif_{i+1} = \actif_i^0$ and that
  $\lost \neq \emptyset$. Note that $\lost \subseteq \actif_i^1$ and
  since the ring is $\myUB$-scattered, this set contains at most
  $\myUB-1$ processes.  By induction hypothesis (\ref{IH-2}), all
  processes between $p'$ and $p$ that are not in $\lost$ are inactive
  and CW-balanced at time $t$. Suppose that there exists $q \in \lost$
  that is not inactive or not CW-balanced at time $t$. Again, assume
  that there is no such process between $q$ and $p$. We know that $q$
  has exited the loop at Lines 4--8 at round $i$ and by (\ref{IH-4}),
  we know that at time $t$, either $q$ is executing the loop at Lines
  17--22 or that it has exited this loop because it has received a
  counter-clockwise message. In the first case, $q$ is active and
  CW-unbalanced at time $t$. In the second case, $q$ is inactive and
  CW-unbalanced by our choice of $q$. In both cases, at time $t$, $q$
  has received at most $2i \myUB - 1$ clockwise messages. By our
  choice of $q$, all processes between $q$ and $p$ are
  CW-balanced. Consequently, by Observation~\ref{obs:diff-recv}, $p$
  has received at most $2i \myUB$ clockwise messages at time
  $t$. But this is impossible as we assumed that at time $t$, $p$ has
  received $(2i+1)\myUB$ clockwise messages. This shows that
  (\ref{IH-2bis}) holds at round $i+1$.
\end{proof}

\begin{lemma}\label{lem-first-loop}
  Every process $p \in \actif_{i}$ exits the loop at Lines 4--8 at
  round $i$.
\end{lemma}

\begin{proof}
  Suppose that there exists a process in $\actif_{i}$ that does not
  exit the loop at Lines 4--8 at round $i$. Then among all such
  processes, consider a process $p$ that has received the minimum
  number $k$ of clockwise messages at the end of the execution. If $p$
  is not exiting the loop, $p$ is CW-unbalanced and $k < (2i-1) \myUB$
  clockwise messages.  Let $p' = \pred_{i}(p)$ and note that by our
  choice of $p$, $p'$ has sent at least $k+1$ clockwise messages. If
  $i = 1$, there are no processes between $p'$ and $p$ and thus there
  is a message in transit between $p'$ and $p$ that will eventually be
  delivered to $p$, a contradiction. Suppose now that $i \geq 2$ and
  that the induction hypothesis holds for round $i-1$.  By induction
  hypothesis (\ref{IH-3}), no process waits for a counter-clockwise
  message at round $j \leq i-1$. Consequently, no process between $p'$
  and $p$ is blocking clockwise messages. Therefore, by
  Observation~\ref{obs:diff-recv}, the number of messages in transit
  is strictly greater than the number of CW-unbalanced processes
  between $p'$ and $p$. Since each such CW-unbalanced process kills
  one clockwise message and then relays the other clockwise messages,
  eventually one clockwise message is delivered to $p$, a
  contradiction.
\end{proof}

\begin{lemma}[\ref{IH-3},\ref{IH-3bis},\ref{IH-4}]
  If $\actif_{i}^0 \neq \emptyset$, then every process
  $p \in \actif_{i}^0$ eventually executes Lines 10 and 11 and the
  message killed by $p$ at Line 11 (at round $i$) was created by
  $\suc_i(p)$ at Line 10 (at round $i$).  Moreover, every process in
  $\actif_{i}^1$ exits the loop at Lines 17--22 with $q = 1$.
\end{lemma}

\begin{proof}
  By Lemma~\ref{lem-first-loop}, every process $p \in \actif_{i}$
  exits the loop at Lines 4--8. The processes in $\actif_{i}^0$ then
  send a counter-clockwise message at Line 10 while the processes in
  $\actif_{i}^1$ start executing the loop at Lines 17--22. 

  Consider two consecutive processes $p, p'' \in \actif_{i}^0$, i.e.,
  such that $p'' = \suc_{i+1}(p)$. Since the ring is
  $\myUB$-scattered, there are at most $\myUB-1$ processes in
  $\actif_{i}^1$ between $p$ and $p''$. Since $p$ has sent
  $(2i-1)\myUB$ clockwise messages and does not send any other
  clockwise message before executing Line 11 at round $i+1$, every
  process from $\actif_{i}^1$ that is between $p$ and $p''$ cannot
  receive more than $(2i-1)\myUB + \myUB -1 < 2i\myUB$ clockwise
  messages before $p$ receives a counter-clockwise message at Line
  11. Therefore, no process from $\actif_{i}^1$ that is between $p$
  and $p''$ can exit the loop at Lines 17--22 with $q=1$ before $p$
  has received a counter-clockwise message.

  By induction hypothesis (\ref{IH-3bis}), this message cannot have
  been created before round $i$. Consequently, no process between $p$
  and $p''$ can have sent this message and thus this message $m$
  has been sent by $p''$ at Line 10 (at round $i$).  Consequently, for
  each process $p^*$ from $\actif_{i}^1$ that is between $p$ and
  $p''$, $p^*$ receives $m$ at Line 19, $p^*$ exits the loop at
  Line 22 with $q=0$ and $p^*$ relay $m$ at Line 24. This shows
  that (\ref{IH-3}), (\ref{IH-3bis}), and (\ref{IH-4}) hold.
\end{proof}

\begin{lemma}[\ref{IH-4bis}]
  If $\actif_{i+1}^0 = \emptyset$, then all processes from
  $\actif_{i+2} = \actif_{i+1} = \actif_{i+1}^1$ exit the loop at
  Lines 17--22 with $q=1$. 
\end{lemma}

\begin{proof}
  The proof uses the same arguments as the proof of
  Lemma~\ref{lem-first-loop}.
\end{proof}

\begin{lemma}[\ref{IH-5}]\label{lem-minleader}
  If a process executes Line 13, then this process is the process
  $p_{0}$ with minimal identifier $\myID_{\min}$, its current round is
  $i_{0} = \len(\myID_{\min})$, and it is the only active process
  reaching this round .
\end{lemma}

\begin{proof}
  Recall that we have assumed that the distribution of identifiers is
  $0$-ended and strongly prefix-free.  For each process $p_1$ with
  identifier $\myID_1$, let $i_1$ be the first bit where
  $\myID_{\min}$ and $\myID_1$ differ. By our assumption on the
  identifiers, we have that $p_{0} \in \actif_{i_1+1}$ while
  $p_1 \in \actif_{i_1}\setminus \actif_{i_1+1}$. Since we have also
  assumed that all identifiers ended with bit $0$, we know that
  $i_1 < i_0$. Consequently, $p_1$ is not active at round
  $i_0$ by (\ref{IH-1}).

  Consequently, $\actif_{i_0} = \{p_{0}\}$ and thus,
  $p_{0}$ is the only process that executes Line 13.
\end{proof}

\begin{lemma}[\ref{IH-6}]\label{lemma:leaderfirst}
  If a process terminates (i.e., executes Line 15 or 26), then a
  process executed Line 13.
\end{lemma}

\begin{proof}
  Consider the first process $p$ that terminates. If $p$ terminates at
  Line 15, it has executed Line 13 before and we are done. Suppose now
  that $p$ terminates at Line 26. This means that $p$ has received two
  consecutive counter-clockwise messages $m_1$ and $m_2$ that arrived
  on port $0$ while being inactive and executing \relay. Let $t_1$ and
  $t_2$ be the moments where $p$ receives respectively $m_1$ and
  $m_2$. Suppose that $m_1$ and $m_2$ were respectively generated by
  some active processes at rounds $i_1,i_2$ such that $i_1 \leq i_2$.

  By (\ref{IH-3bis}), $p$ can receive at most one counter-clockwise
  message that has been generated by an active process at Line 10 at
  round $i_1$.  Suppose that no process has executed Line 13 before
  $p$ receives $m_2$, then necessarily $i_1 < i_2$. Let
  $p''_1 \in \actif_{i_1}^0$ (respectively,
  $p''_2 \in \actif_{i_2}^0$) be the process that creates message
  $m_1$ (respectively $m_2$) by executing Line 10 at round $i_1$
  (respectively, $i_2$). Let $p'_1 \in \actif_{i_1}^0$ (respectively,
  $p'_2 \in \actif_{i_2}^0$) be the process that kills message $m_1$
  (respectively $m_2$) by executing Line 11 at round $i_1$
  (respectively, $i_2$). By (\ref{IH-3bis}),
  $p'_1 = \pred_{i_1}(p_1'')$ and $p$ appears between $p_1'$ and
  $p_1''$. Similarly, $p'_2 = \pred_{i_2}(p_2'')$ and $p$ appears
  between $p_2'$ and $p_2''$. Moreover, either $p'_1 = p'_2$
  (respectively, $p''_1 = p''_2$) or $p'_1$ (respectively, $p''_1$) is
  between $p'_2$ and $p_2''$. 

  Note that when $p$ receives $m_1$, $p'_1$ has not finished executing
  round $i_1$ and $p_1$ has thus sent at most $(2i_1 - 1)\myUB$
  clockwise messages. Consequently, by (\ref{IH-4}), no process from
  $\actif_{i_1}^1$ has received more than $2i_1\myUB -1$ clockwise
  messages at time $t_1$. Consequently, no process between $p_1'$ and
  $p_1''$ has received more than $2i_1\myUB$ clockwise messages at
  time $t_1$. When $p''_2$ sends $m_2$, it has received
  $(2i_2 - 1)\myUB$ clockwise messages. By (\ref{IH-2}) all processes
  between $p'_2$ and $p''_2$ are CW-balanced. Consequently, every
  process between $p'_2$ and $p''_2$ has received at least
  $(2i_2 - 1)\myUB$ clockwise messages at time $t_2$. Since
  $i_2 > i_1$, we have that $(2i_2 - 1)\myUB > 2i_1
  \myUB$. Consequently, $p$ has necessarily received a clockwise
  message between $m_1$ and $m_2$, but this contradicts the definition
  of $m_1$ and $m_2$.
\end{proof}
\color{black}
We can now prove that Algorithm~\ref{algo:orientedknown} is a correct
leader election algorithm. 

\begin{proposition}\label{prop-alg-main}
  Algorithm~\ref{algo:orientedknown} is a quiescently terminating
  non-uniform leader election algorithm for oriented rings with a
  distribution of identifiers that is $0$-ended, strongly prefix-free
  and $\myUB$-scattered. It elects the process $p_0$ with the minimum identifier
  $\myID_{\min}$ and during its execution, each process sends
  $(2\len(\myID_{\min})-1)\myUB = O (\myUB \log (\myID_{\min}))$
  clockwise messages and at most
  $\len(\myID_{\min})+1 = O(\log (\myID_{\min}))$ counter-clockwise
  messages.
\end{proposition}

\begin{proof}
  By (\ref{IH-5}) and (\ref{IH-6}), as long as $p_0$ has not executed
  Line 13, no process terminates. Moreover, by (\ref{IH-1}), $p_0$ is
  the only process executing round $i_{0} = \len(\myID_{\min})$ and
  by~(\ref{IH-3}), it eventually executes Lines 10 and 11 at round
  $i_0$. By~(\ref{IH-2}), since $\pred_{i_0}(p_0) = p_0$, when $p_0$ executes
  Line 10, all other processes are inactive and CW-balanced. By
  (\ref{IH-3bis}), at this point, there are no counter-clockwise
  messages in transit. Consequently, when the counter-clockwise
  message $m$ sent by $p_0$ at line 10 at round $i_0$ is received by
  $p_0$ at Line 11, all other processes have relayed $m$. Then $p_0$
  send another counter-clockwise message $m'$ at Line 13 that is also
  relayed by the other processes. When a process $p$ receives the
  message $m'$, the last two messages it received were $m$ and $m'$
  and thus, at this point, $p$ forwards $m'$ and returns
  \nonleader. When $m'$ is received by $p_0$ at Line 14, $p_0$ returns
  $\leader$ and at this point, there are no more messages in transit
  and all other processes have terminated the algorithm, returning
  \nonleader. Consequently, our algorithm is a leader election
  algorithm with quiescent termination.

  During the execution, there are $i_0$ rounds. At the end of the
  execution, $p_0$ has sent and received exactly $(2 i_0 -1)U$
  clockwise messages (recall that the last bit of $\myID_{\min}$ is
  $0$). Since there are no more messages in transit and since all
  processes are CW-balanced at the end of the algorithm, by
  Observation~\ref{obs:diff-recv}, all processes have sent and
  received exactly $(2 i_0 -1)U$ clockwise messages. The number of
  counter-clockwise messages sent by $p_0$ is precisely $1$ plus the
  number of bits in $\myID_{\min}$ that are equal to $0$ (i.e., the
  rounds for which $\bit(\myID_{\min},i)=0$). By~(\ref{IH-3bis}), all
  processes have sent and received this same number of
  counter-clockwise messages.
\end{proof}

The proof of the main theorem follows from the
previous proposition since we can ensure that each distribution of
identifiers is $0$-ended, strongly prefix-free and $U$-scattered.

\algtwo*

\section{Randomized Leader Election} \label{sect: rand}

In this section, we prove our randomized main result.
\mainthmrand*

\begin{algorithm}
\DontPrintSemicolon
\caption{RandomizedLeaderElection($U$)}
\label{alg: rand oriented ring}
$\rand_v \sim \mathsf{UniformSampling}(0,2^{\lceil c_1\log U\rceil}-1)$ \;
$\id_v \leftarrow 2*(2^{\lceil c_1\log U\rceil} + \rand_v)$ \;
$\mathsf{LogarithmicOrientedLeaderElection}(\id_v,{\lceil c_2 \log U\rceil})$\;
\end{algorithm}

In this section, we use $U(a,b)$ (or $\mathsf{UniformSampling}(a,b)$ in the algorithm description) to denote the uniform distribution over the set of integers $\mathbb{Z} \cap [a,b]$. 

In our algorithm, each process first draws a number uniformly at
random within the interval $[0,2^{\lceil c_1\log U\rceil}-1]$. Observe
that all these numbers are encoded in at most
$\lceil c_1\log U\rceil - 1$ bits.  Then, each process adds
$2^{\lceil c_1\log U\rceil}$ to its number to obtain a number that
is encoded on precisely $\lceil c_1\log U\rceil$ bits. Then, they
multiply this number by $2$ to obtain a number of fixed length that
ends with a $0$. Each node then takes this new number as its
identifier.  This ensures that the distribution of identifiers is
strongly prefix-free and $0$-ended. Then, the processes run
Algorithm~\ref{algo:orientedknown} with
$\myUB = \lceil c_2 \log U \rceil$.

We first argue that since each node samples $\rand_v$ from the
discrete uniform distribution $U(0,2^{\lceil c_1\log U\rceil}-1)$,
with high probability, no two processes share the same identifier.


\begin{lemma} [Unique identifiers w.h.p.]
\label{lem: unique max whp}
Assume there exist $n$ nodes. 
With probability at least $1-U^{2-c_1}$, all nodes have
different identifiers .
\end{lemma}

\begin{proof}
  Note that two processes $u,v$ have the same identifier if and only
  if $\rand_u = \rand_v$.
    \begin{align*}
                      \Pr[\text{No $u,v$ hold identical $\id$}]
                    &\geq 1 - \sum_{u\neq v} \Pr[\rand_u = \rand_v] \\
                    &\geq 1- \frac{U^2}{2^{\lceil c_1\log U\rceil}} \\
                    &\geq 1-U^{2-c_1}.\qedhere
    \end{align*}
\end{proof}
\color{black}

Now we proceed to show that our choice of
$\myUB = \lceil c_2\log U \rceil$ is justified, i.e., we need to show
that with high probability, the distribution of identifiers is
$\myUB$-scattered.  To that end, we need to show that with high
probability, at each round $i$ in the execution, among the active
processes $A_i$ at round $i$, either $A_i = A_i^1$ (i.e., all
processes in $A_i$ have a bit $1$ at position $i$), or there are at
most $\myUB-1$ consecutive nodes in $A_i^1$. Note that all
identifiers have the same length $\ell = \lceil c_1\log U\rceil+1$,
they all start with $1$ and end with $0$. Thus, we only need to
consider the bits that are different from the first and the last bits.



\begin{lemma}[$\lceil c_2\log U\rceil$-scatteredness
  w.h.p.]\label{lem: run freeness whp}
  For each bit $2 \leq i \leq \ell -1 = \lceil c_1\log U\rceil$, with
  probability at least $1 - U^{1-c_2}$, among $A_i$, there exists at
  most $\myUB -1$ consecutive nodes in $A_i^1$.
\end{lemma}

\begin{proof}
  Since for each process $v$, $\rand_v$ is a number drawn uniformly at
  random between $0$ and $2^{\lceil c_1\log U\rceil}-1$, it can be
  seen as equivalent to sampling $\lceil c_1\log U\rceil$ bits
  independently from $U(0,1)$. 
%
%
  By the independence of choices, revealing, or conditioning on the
  first $i-1$ bits does not change the distribution of the $i$-th
  bit. This implies that the $i$-th bits of all processes in $A_i$ are
  sampled independently from $U(0,1)$.

  In the below analysis, we assume that $A_i$ has at least
  $\lceil c_2\log U\rceil$ nodes.  Otherwise the lemma is trivially
  true for $A_i$.  Now we define an event $R_i(v)$ for every
  $v \in A_i$ as ``In $A_i$, starting from $v$, there are
  $\lceil c_2\log U\rceil$ consecutive nodes from $A_i^1$ when moving
  in the clockwise direction''. Immediately,
  $\Pr[R_i(v)] = 2^{-\lceil c_2\log U\rceil}$, given that each label
  is independently sampled from $U(0,1)$. On the other hand, the lemma
  holds for $A_i$ if and only if $R_i(v)$ does not happen for any
  $v$. We hence make use of the union bound.
    \begin{align*}
        \Pr[\text{$A_i$ contains at most $\myUB -1$ consecutive nodes}] 
                &\geq 1-\sum_{v \in A_i} \Pr[R_i(v)] \\
                &\geq 1-\frac{U}{2^{\lceil c_2\log U\rceil}} \\
                &\geq 1-U^{1-c_2}\qedhere
    \end{align*}
\end{proof}

We are now ready to prove the main theorem.

\mainthmrand*

\begin{proof}
  By Proposition~\ref{prop-alg-main}, the randomized algorithm is
  successful provided each node has a unique identifier, and the
  distribution of identifiers is $0$-ended, strongly prefix-free and
  $\lceil c_2\log U\rceil$-scattered. Note that the distribution of
  identifiers is always $0$-ended and strongly prefix-free. The other
  two requirements only hold w.h.p.: we want to bound the small
  failure probability that (1) the identifiers are not unique and (2)
  at some round $2 \leq i \leq \lceil c_1\log U\rceil$, in $A_i$ there
  are $\lceil c_2\log U\rceil$ consecutive processes in $A_i^1$.
   

  By \Cref{lem: unique max whp}, with probability at most $U^{2-c_1}$,
  there exists two processes $u,v$ that have the same identifier. For
  each $2 \leq i \leq \lceil c_1\log U\rceil$, by \Cref{lem: run
    freeness whp}, with probability at most $U^{1-c_2}$, $A_i$
  contains $\lceil c_2\log U\rceil$ consecutive processes in $A_i^1$.
  Over $\lceil c_1 \log U\rceil -1$ such rounds $i$, the probability
  that some single bit ring does not meet the requirement is capped at
  $\lceil c_1 \log U\rceil \cdot U^{1-c_2} = O(U^{2-c_2})$. By union
  bound, the total failure probability is capped by
  $O(U^{2-c_2}) + U^{2-c_1} = O(U^{2-\min(c_1,c_2)})$. For a success
  probability of $1-U^{-c}$, we only need to choose $c_1 \geq c+2$ and
  $c_2 \geq c+2$.

  By Proposition~\ref{prop-alg-main}, in the execution of the
  algorithm, each process sends
  $O(\myUB \log \myID_{\min}) = O( c_2 \log U \cdot c_1 \log U)$
  messages. Consequently, the total message complexity is
  $O(n \cdot \log^2 U)$.
\end{proof}

\subparagraph{Discussion.}  An interesting question is whether the
message complexity can be further improved by employing a different
$\id$ generator than the uniform distribution, with the goal of
producing shorter $\id$s. If the proof of
Proposition~\ref{prop-alg-main}, one can observe that we do not need
that the identifiers are all different, but only that there is a
unique process with the minimum one. 
In this sense, using a uniform distribution may be unnecessarily
strong, and one can then hope to produce shorter identifiers with a unique minimum in order to decrease  the message complexity.

\section{Constant Messages in One Direction: Uniform vs Non-Uniform}\label{sec:knowingn}

In this section, we investigate the possibility of designing leader election algorithms that send only a constant number of messages in a given direction. We first show that no uniform algorithm can achieve this. Then, we demonstrate that, when the algorithm is non-uniform, it is possible to elect a leader by sending just three messages counter-clockwise.

\subsection{Uniform Algorithms: No Constant Messages in a Single Direction}\label{sec:imp}
 We show that no uniform algorithm can send a constant number of messages in a single direction; without loss of generality, the counter-clockwise direction.

\boundth*

In the following, we consider a fixed uniform leader election algorithm $\cA$ for oriented rings. Without loss of generality, $\cA$ is an event-driven algorithm: initially, the algorithm starts with an initialization event and subsequently reacts to delivery events. When it reacts to an event, it is able to send messages in any direction.
Notice that such formalism to represent the algorithm is different from the one that we used in the algorithms that we propose in our paper. However, this difference is only a presentation difference and does not restrict the validity of our result. 

We assume that ${\cA}$ violates the hypothesis of our theorem, that
is, each process executing ${\cA}$ sends fewer than $b$
counter-clockwise messages on any ring, regardless of its identity,
the identities of the other processes, or the execution.

We now show, by contradiction, that ${\cA}$ cannot exist. As
in~\cite{frei2024content}, we define a solitude pattern for each identifier by
considering an execution of the algorithm by a single process with
this identifier. Our definition of the solitude pattern associated to
each identifier is however slightly different from the one of~\cite{frei2024content}.

\begin{definition}
Consider an execution of ${\cA}$ on an oriented ring with a single process ($n = 1$) with identifier $\myID$, where messages are delivered one by one prioritizing the delivery of clockwise messages (i.e., the process receives a counter-clockwise message only when all sent clockwise messages have been received). We call this execution the solitary execution of the process with identifier $\myID$. 

Let $t(\myID)$ be the number of counter-clockwise messages received
before the termination of this execution of ${\cA}$. The solitude
pattern of $\myID$ is the sequence
$\sol(\myID)=(cw_0(\myID),cw_1(\myID), \ldots,$
$ cw_{t(\myID)}(\myID))$ where $cw_{i}(\myID)$ (with $i < t(\myID)$)
is the number of clockwise messages received by the process before
receiving the $(i+1)$-th counter-clockwise message, and
$cw_{t(\myID)}(\myID)$ is the number of clockwise messages received
before termination. When $\myID$ is clear from context, we  also
write $t$ for $t(\myID)$ and $cw_i$ for $cw_i(\myID)$.
\end{definition}

We remark that $cw_i(\myID)$ is the total number of clockwise messages
received by the process from the beginning up to the receipt of the
$(i+1)$-th counter-clockwise message. This implies that each solitude
pattern is a non-decreasing sequence.

Since in a solitary execution a process receives only messages sent by itself, it must have sent at least $cw_{i}$ clockwise messages before receiving the $(i+1)$-th counter-clockwise message. Moreover, by the end of the execution, it has sent at least $t$ counter-clockwise messages and $cw_{t}$ clockwise messages.

It is also easy to observe that, at the end of a solitary execution, the process should elect itself as the leader, as if it is on a ring of size $1$.

Note that the solitude pattern $\sol(\myID)$ of $\cA$ depends only on
$\myID$ and $\cA$. Moreover, it is well-defined (that is, $t(\myID)$
and all $cw_i(\myID)$ are finite) since $\cA$ must terminate on a ring
of size $1$, where there is a single process with identifier
$\myID$. By our assumption, $t \leq b$, which implies that there
exists an infinite number of identifiers that have solitude patterns
of the same length. In the following, we assume that this length is
$k+1$, i.e., in the solitary execution of each of these identifier,
the unique process sends $k$ counter-clockwise messages .

As observed by \cite[Lemma 19]{frei2024content}, two distinct identifiers have
different solitude patterns, since otherwise the algorithm fails to elect a leader in a ring with two processes.


\begin{lemma}\label{lemma:particular-infinite-set}
  There exists an index $0 \leq \ell \leq k$ and an infinite
  sequence $(\myID_j)_{j \in \N}$ of identifiers such that (1)
  $cw_{\ell}(\myID_{j+1}) > cw_{k}(\myID_{j})$, for every $j \in \N$,
  and (2) for each $0 \leq i < \ell$,
    $cw_i(\myID_{j}) = cw_i(\myID_{j'})$ for all $j,j' \in \N$.
\end{lemma}

\begin{proof}
  We prove by induction on $k$ that in any infinite set $U$ of
  distinct $(k+1)$-tuples $a = (a_0, a_1, \ldots, a_k)$, there exists an
  index $0 \leq \ell \leq k$, integers
  $c_0 \leq c_1 \leq \ldots \leq c_{\ell-1}$, and an infinite sequence
  $(a_j)_{j \in \N}$ with $a_j = (a_{j,0}, a_{j,1}, \ldots, a_{j,k})$
  such that for every $j \in \N$, (1) $a_{j+1,\ell} > a_{j,k}$, and
  (2) $a_{j,i} = c_i$ if $i < \ell$.
  
  Let $U_0 = \{a_0 \mid \exists (a_0, a_1, \ldots, a_k) \in U\}$ be
  the projection of $U$ on the first coordinate. We distinguish two
  cases depending on whether $U_0$ is finite or not. Note that if
  $k = 0$, $U_0 = U$ is necessarily infinite.

  Suppose first that $U_0$ is infinite, i.e., for any value
  $c \in \N$, there exists $(a_0, a_1, \ldots, a_k) \in U$ such that
  $a_0 > c$. In this case, we construct the sequence
  $(a_j)_{j \in \N}$ iteratively as follows. Start with $S_0 = U$ and
  for each $j \in \N$ , let $a_j = (a_{j,0}, \ldots, a_{j,k})$ be an
  arbitrary element from $S_j$ and let
  $S_{j+1} = \{a' = (a_0',\ldots, a_k') \in S_j \mid a_0' >
  a_{j,k}\}$. Since $U_0$ is infinite, we know that each $S_{j+1}$ is
  non-empty (and infinite).  Consequently, $(a_j)_{j \in \N}$ is
  well-defined and satisfies the conditions with $\ell = 0$.

  Suppose now that $U_0$ is finite. Then, by the pigeonhole principle,
  there exists an infinite subset $U'$ of $U$ and a value $c_0$ such
  that for every $(a_0, a_1, \ldots, a_k) \in U'$, $a_0 = c_0$.  Now
  consider the set
  $U'' = \{a''=(a_1, \ldots, a_k) \mid \exists a=(a_0=c_0, a_1, \ldots,
  a_k) \in U'\}$. Since $U'$ is infinite and since all elements of
  $U'$ have the same first coordinate, necessarily $U''$ is an
  infinite set of distinct $k$-tuples. By the inductive
  hypothesis, there exists $1 \leq \ell \leq k$, integers
  $c_1\leq \ldots \leq c_{\ell-1}$, and an infinite sequence
  $(a''_j)_{j \in \N}$ with
  $a''_j = (a''_{j,1}, a''_{j,2}, \ldots, a''_{j,k})$ such that for
  every $j \in \N$:, (1) $a'' _{j+1,\ell} > a''_{j,k}$, and (2)
  $a''_{j,i} = c_i$ if $1 \leq i < \ell$.
  Consequently, the sequence $(a_j)_{j \in \N}$ where
  $a_j = (c_0, a''_{j,1}, a''_{j,2}, \ldots, a''_{j,k})$ for each
  $j \in \N$ is a sequence of elements of $U' \subseteq U$ satisfying
  the condition.
\end{proof}
\color{black}



We are now ready to prove our theorem
\boundth*

\begin{proof}
  Suppose there exists an algorithm $\cA$ for oriented rings such that
  each process executing $\cA$ sends fewer than $b$ counter-clockwise
  messages, regardless of the ring size, its identity, or the
  execution. We will construct a ring and an execution in which all
  processes execute $\cA$, but two of them will be unable to
  distinguish this execution from their respective solitary
  executions. As a result, both will enter the leader state and
  terminate, violating correctness.

  \begin{figure}
    \centering
    \includegraphics[scale=0.8]{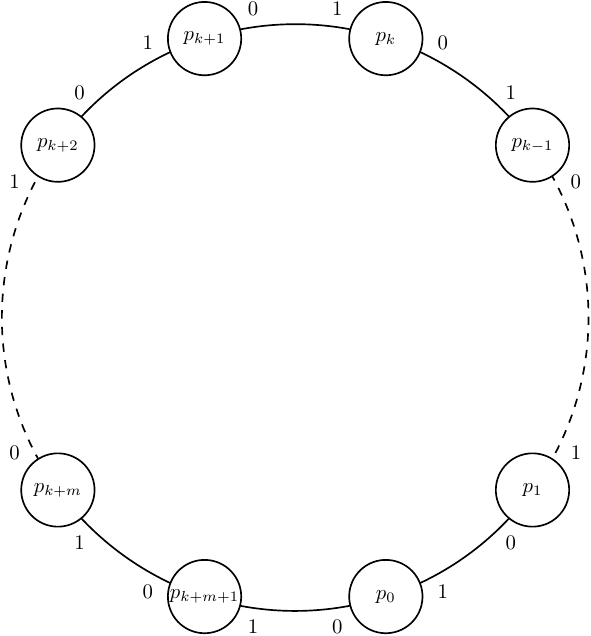}
    \caption{Arrangement of processes in our impossibility proof}
    \label{fig:enter-label}
  \end{figure}

  By Lemma~\ref{lemma:particular-infinite-set}, there exists an
  integer $0 \leq \ell \leq k$, a sequence of integers
  $cw_0 \leq cw_1 \leq \ldots \leq cw_{\ell-1}$, and an infinite
  sequence of identifiers $(\myID_j)_{j \in \N}$ such that for each
  $j \in \N$, (1) $cw_{\ell}(\myID_{j+1}) > cw_{k}(\myID_{j})$, and
  (2) for each $0 \leq i < \ell$, $cw_i(\myID_{j}) = cw_i$.

  Consider the solitary pattern
  $\sol(\myID_{k+1}) = (cw_0(\myID_{k+1}), cw_1(\myID_{k+1}), \ldots,
  cw_k(\myID_{k+1}))$ of identifier $\myID_{k+1}$ in the sequence and
  let $m = cw_k(\myID_{k+1})$.

  We consider $k + m + 2$ processes $p_0, \ldots, p_{k+m+1}$ with
  identifiers $\myID_0, \ldots, \myID_{k+m+1}$.

  The processes are
  arranged on the ring starting from $p_0$, so that each $p_i$ is the
  counter-clockwise neighbor of $p_{i-1}$, and $p_0$ is the
  counter-clockwise neighbor of $p_{k+m+1}$ (see
  Figure~\ref{fig:enter-label}).

  We now construct an execution where $p_k$ and $p_{k+1}$ behave as in
  their solitary execution. The execution consists of three phases
  (the first does not exist if $\ell = 0$).  For each
  $0 \leq i < \ell$, we let $cw_i = cw_i(\myID_0)$ and observe that by
  the definition of $\ell$, $cw_i = cw_i(\myID_j)$ for any $j \in
  \N$. Now we construct an execution where all processes reach the
  state they reach in a solitary execution when they have received
  $cw_{\ell-1}$ clockwise messages and $\ell$ counter-clockwise
  messages. During this phase, all processes receives the same number
  of messages from both directions in the same order.

  Assume that the processes have received $t$ clockwise messages and
  $i$ counter-clockwise messages, that $t \leq cw_{i}$, and that
  $cw_{i-1} \leq t$ if $i \geq 1$. Moreover, we assume that the
  processes have received the messages in the same order as in the
  solitary execution (and are thus in the same state as if they were
  alone after having received these messages).  Initially, $t = 0$ and
  $i = 0$ and the conditions are satisfied.  Suppose first that
  $t < cw_i$. Then, in the solitary execution of every identifier
  $\myID_j$ of a process on the ring, there is at least one clockwise
  message that has been sent and not yet received. Since the processes
  in the ring have behaved up to now as in their solitary executions,
  there is at least one clockwise message in transit between any pair
  of consecutive processes. So, we deliver a clockwise message to
  every process and let them react as in their solitary
  execution. Note that the conditions are still satisfied. Suppose now
  that $t = cw_{i}$. Then, in the solitary execution of every
  identifier $\myID_j$, there are no more clockwise messages in
  transit, the process has sent at least $i$ counter-clockwise
  messages and has received $i-1$ counter-clockwise
  messages. Therefore, in the solitary execution of every identifier
  $\myID_j$, there is a counter-clockwise message that has been sent
  and not yet received. For the same reasons as above, there is at
  least one counter-clockwise message in transit on every link of the
  ring. So we deliver a counter-clockwise message to every process and
  let them react as in their solitary execution.  By this reasoning,
  we can assume that we reach a configuration where all processes have
  sent at least $cw_{\ell-1}$ clockwise messages and at least $\ell$
  counter-clockwise messages, and received exactly $cw_{\ell-1}$
  clockwise messages and $\ell$ counter-clockwise messages.

  We now enter in the second phase of the execution. In this phase, we
  only deliver clockwise messages. At this point, any process $p_j$
  that has received $s$ clockwise messages with
  $s < cw_{\ell}(\myID_j)$ has sent at least $s+1$ clockwise
  messages. As long as there is a process $p_j$ that has not yet
  received $cw_{\ell}(\myID_j)$ clockwise messages and such that there
  is a clockwise message in transit from $p_{j+1}$ to $p_j$, we
  deliver this message to $p_j$. The second phase ends when for each
  $0 \leq j \leq k+m+1$, $p_j$ has received $cw_{\ell}(\myID_j)$
  clockwise messages or there is no more clockwise messages in transit
  between $p_{j+1}$ and $p_j$.

  \begin{claim}\label{claim:flooding-msg}
    At the end of the second phase of the execution, the following holds:
    \begin{enumerate}[label=(\arabic*)]
    \item for any $k+1 < j \leq m+k+1$, $p_j$ has received at least
      $m+k+1-j$ clockwise messages and send at least $m+k+2-j$ clockwise
      messages
    \item for any $0 \leq j \leq k+1$, $p_j$ has received and sent
      $cw_{\ell}(\myID_j)$ clockwise messages, and there are at least
      $cw_{k}(\myID_j) - cw_{\ell}(\myID_j)$ clockwise messages in
      transit between $p_{j+1}$ and $p_j$.
    \end{enumerate}
  \end{claim}

  \begin{proof}
    We prove the claim by reverse induction on $j$.  Note that for
    $j = m+k+1$, the claim says that $p_{m+k+1}$ has received at least
    $0$ clockwise messages and sent at least $1$ clockwise
    message. This property holds since
    $cw_{\ell}(\myID_{m+k+1}) > cw_{k}(\myID_{m+k}) \geq 0$. 
    
    Suppose now that $k+1 < j< m+k+1 $, then by induction hypothesis,
    $p_{j+1}$ has sent at least $m+k+2-j-1 = m + k + 1 -j < m$
    messages. Since
    $cw_{\ell}(\myID_j) \geq cw_{k}(\myID_{k+1})+(j -(k+1)) > m$, the
    scheduler has delivered at least $m + k + 1 -j$ clockwise messages
    sent by $p_{j+2}$ to $p_{j+1}$ and thus the claim holds for
    $j$.

    If $j = k+1$, then by induction hypothesis, $p_{k+2}$ has sent at
    least $m = cw_k(\myID_{k+1})$ clockwise messages to $p_{k+1}$,
    among which $cw_{\ell}(\myID_{k+1})$ have been delivered to
    $p_{k+1}$. Thus, $p_{k+1}$ has also sent $cw_{\ell}(\myID_{k+1})$
    clockwise messages, and there are at least
    $cw_k(\myID_{k+1}) - cw_{\ell}(\myID_{k+1})$ messages in transit
    between $p_{k+2}$ and $p_{k+1}$. 

    Suppose now that $0 \leq j < k+1$. Then $p_{j+1}$ has sent at
    least $cw_{\ell}(\myID_{j+1}) > cw_{k}(\myID_{j})$ clockwise
    messages to $p_j$, among which $cw_{\ell}(\myID_{j})$ have been
    delivered to $p_j$. Thus, $p_j$ has also sent
    $cw_{\ell}(\myID_{j})$ clockwise messages, and there are at least
    $cw_k(\myID_{j}) - cw_{\ell}(\myID_{j})$ messages in transit
    between $p_{j+1}$ and $p_{j}$. This ends the proof of the claim.
  \end{proof}
  \color{black}
  We now describe the third phase of the execution that is made of
  $k-\ell$ rounds numbered from $\ell+1$ to $k$. In this phase, not
  all processes are active: at round $i$, only the processes $p_j$
  with $i \leq j \leq k+1$ are active. We show by induction on $i$
  that at the beginning of round $i$ (resp., at the end of round $i$),
  each process $p_j$ with $i-1 \leq j \leq k+1$ (resp., with
  $i \leq j \leq k+1$) has performed the exact same step as in the
  solitary execution of $\myID_j$ until it has received $i-1$ (resp.,
  $i$) counter-clockwise messages and $cw_{i-1}$ (resp., $cw_i$)
  clockwise messages. Observe that the execution constructed
  so far satisfy this property at the beginning of round
  $i = \ell +1$.


  Assume that the invariant holds at the beginning of round $i$ and
  consider $p_j$ with $i-1 \leq j \leq k+1$. Since $p_j$ behaves as in
  its solitary execution, since $p_j$ has received $i-1$
  counter-clockwise messages and $cw_{i-1}$ clockwise messages, and
  since $i \leq k$, we know that $p_j$ has sent sent at least $i$
  counter-clockwise messages. Consequently, there is a
  counter-clockwise message in transit between $p_j$ and $p_{j+1}$ for
  any $i-1 \leq j \leq k$, and we can deliver a counter-clockwise
  message to any process $p_j$ with $i \leq j \leq k+1$. Then, we
  deliver clockwise messages to the processes $p_j$ such that
  $i \leq j \leq k+1$ until each of them has received $cw_i(\myID_j)$
  clockwise messages. This is possible since by
  Claim~\ref{claim:flooding-msg}, there are enough clockwise messages
  in transit between $p_{j+1}$ and $p_j$. This ends round $i$.

  At this point, we have reached a configuration where each process
  $p_j$ with $i \leq j \leq k+1$ has performed the same steps as in
  the solitary execution of $\myID_{j}$, $p_j$ has sent and receive
  precisely $cw_{i}(\myID_{j})$ clockwise messages and it has received
  $i$ counter-clockwise messages. Consequently, the invariant is
  preserved at the end of round $i$ (and at the beginning of round
  $i+1$ when $i < k$).


  When we end round $i = k$, then both $p_k$ and $p_{k+1}$ are in the
  same state as their final states in the solitary executions of
  $\myID_k$ and $\myID_{k+1}$, i.e., they are both in the \leader
  state. Consequently, we have constructed a ring and an execution of
  $\cA$ on this ring where two processes terminate their execution in
  the \leader state. This shows that $\cA$ cannot be a uniform
  leader election algorithm and ends the proof of the theorem.
\end{proof}

\subsection{A Constant Direction Non-Uniform Algorithm}\label{sec:cdir}
In this section we describe an algorithm where each process sends $O(U\cdot \myID_{\min})$ messages in the clockwise direction, and just 3 messages in the counter-clockwise direction. 
This algorithm is of interest as in Section \ref{sec:imp} we show that uniform algorithms are not able to elect a leader if they send a bounded number of messages in one direction. The algorithm of this section shows that this limitation no longer holds once the knowledge of an upper bound on $U \geq n$ is given to processes.

The algorithm, pseudocode in Algorithm \ref{algo:c1}, starts in a competing phase (in blue at Lines 1–4). In the competing phase, a process $p$ with identifier $\myID$ sends and receives $U \cdot \myID$ messages. This phase synchronizes the processes in such a way that two indices of processes in the phase loop can differ by at most $n - 1$. The process with the minimum identifier will be the only one to end the competing phase by completing the prescribed number of iterations.

This process then enters the termination detection phase (in green at Lines 14–24). In this phase, the process first sends a message counter-clockwise to inform all other processes that they are not leaders. Then, it sends a second counter-clockwise message to ensure that all clockwise messages still traveling in the network are eliminated. Finally, it sends a third and final counter-clockwise message to terminate the algorithm.

A process whose identifier is not the minimum will enter the relay phase
(in red at Lines 6–12). In this phase, the purpose of the process is
to relay both counter-clockwise and clockwise messages, and to
eliminate the clockwise message it sent at Line 2 (note that a
non-leader process exits the loop upon receiving a counter-clockwise
message at Line 3). A non-leader process terminates after receiving
three counter-clockwise messages: one at Line 3 and the other two
while executing the loop at Lines 9–12.

\begin{algorithm}
\footnotesize
  \DontPrintSemicolon
  \SetKw{Send}{send}
  \SetKw{Receive}{receive}
    \SetKw{Break}{break}
  \SetKw{UnorientedRelayAndWaitTermination}{UnorientedRelayAndWaitTermination}
  \SetKwData{Status}{status}
  \SetKwFunction{Leader}{Leader}
  \SetKwFunction{NonLeader}{Non-Leader}
    \SetKwFunction{Undecided}{Undecided}
    \SetKwData{Count}{count}
    \SetKwData{ID}{ID}
  \SetKwData{SizeNetwork}{boundOnSize}
 \color{blue}
  \For{$i \leftarrow 1$ \KwTo $U \cdot \ID$}
    {

      \Send a message on port $0$\;
      \Receive a message on port $q$\;
        \lIf{$q=0$}{ \Break}
    }
    \color{black}
    \eIf{$q=0$}{
      \color{red}      
      \tcc{if a process enters this branch, then it is not the process
        with minimum $\myID$}
     \Send a message on port $1$ \;
       \Receive a message on port $1$\;
        \Repeat{received two messages on port $0$}
        {
            \Receive a message on port $q$\;
            \Send a message on port $1-q$\;
        }
        \Return{\NonLeader} \;
        \color{black}
      }{
   \tcc{only the process with minimum $\myID$ enters this branch}
   \color{teal}
    \Send a message on port $1$ \;
    \Receive a message on port  $0$\; 
    \Send a message on port $1$ \;
     \Repeat{$q=0$}
    {
             
        \Receive a message on port $q$\;
        \If{$q=1$}{
         \Send a message on port $0$\;
        }
    }
    \Send a message on port $1$ \;
    \Receive a message on port  $0$\; 
   
    \Return{\Leader}     
    \color{black}
    }
 \caption{ElectionWithOneDirectionConstant(\myID,U)}\label{algo:c1}
\end{algorithm}

\subsubsection{Correctness}
We recall the definition of \emph{CW-unbalanced}, that is a process that has sent more messages on port $0$ than it
has received on port $1$. A process that is not CW-unbalanced
is said to be \emph{CW-balanced}.

For convenience and without loss of generality in the following we assume that $p_0$ is the process with the minimum identifier $\myID_0$. 

\begin{lemma}\label{lemma:minid}
All processes eventually exit the loop at Lines 1--4. The condition at Line 5 is false only for process $p_0$. Therefore, the process with the minimum identifier $\myID_0$ eventually executes Line 14 after having sent $U \cdot \myID_0$ clockwise messages at Line 3. All
  other processes eventually execute Line 6 after having sent at most  $U \cdot \myID_0+n$ messages at Line 3.  
\end{lemma}

\begin{proof}
  Observe that a process cannot exit the loop at Lines 1--4 before it
  receives a counter-clockwise message (on port $0$) at Line 4 or it has done $U\cdot \myID$ iterations of
  the loop. In order to receive a counter-clockwise message there must be a process that exited the loop. The first process to exit the loop must do so by completing the number of iterations.

  Note also that as long as no process has exited the loop at Lines
  1--4, a process that is waiting for a message is CW-unbalanced. Consequently, if all $n$ processes are waiting for a message,
  then there are $n$ clockwise messages in transit and at least one of them is
  eventually delivered. Therefore, as long as no process has exited
  the loop at Lines 1--4, the sum of the number of iterations of the
  loop performed by each process is increasing. At some
  point, there is a process that has performed $U\myID$ iterations and
  exits the loop.
  
  Let $p_j$ be the first node exiting the loop.  Suppose, by
  contradiction, that $p_j$ is not the process with the minimum
  identifier $\myID_{0}$, i.e., $\myID_j \geq \myID_{0} +1$. Then,
  when $p_j$ exits the loop at Lines 1--4, it has received at least
  $U\cdot \myID_{0}+U$ clockwise messages (at Line 3) and by Lemma
  \ref{lemma:forwardbound}, each other process has sent at least
  $U\cdot \myID_{0}+U-n+1$ clockwise messages (at Line 2) but this
  would imply that the process $p_0$ with $\myID_{0}$ has started
  iteration $U \cdot \myID_{0}+1$ of the loop which is impossible.
  
  Consequently, when $p_0$, the process with the minimum identifier,
  exits the loop, the condition at Line 5 is false, and thus $p_0$ is
  the first process executing Line 14, sending a counter-clockwise
  message. Note that at this point $p_0$ is CW-balanced.

Once this happens, $p_0$ waits at Line 15 to receive a counter-clockwise message (thus on port $0$). In particular, as long as it is waiting, it stops forwarding clockwise messages. Since $p_0$ is CW-balanced, there are at most $n - 1$ clockwise messages in transit.

Now observe that $p_0$ sent $U \cdot  \myID_{0}$ clockwise messages. Since there are at most $n - 1$ clockwise messages in transit and $p_0$ has stopped forwarding, the other processes can have received at most $U \cdot  \myID_{0} + (n - 1)$ clockwise messages. Thus they cannot have completed iteration $U \cdot (\myID_{0} + 1)$ and therefore they can only exit the loop at Lines 1--4 by receiving a counter-clockwise message. 

Since all processes different from $p_0$ relay (at Line 6) the counter-clockwise message sent by $p_0$ (at Line 14), each process $p \neq p_0$ eventually receives a counter-clockwise message (at Line 3), and therefore executes Line 6 since the condition at Line 5 is satisfied.  At this point by Lemma \ref{lemma:forwardbound} they have sent at most $U \cdot \myID_{0} + n$ clockwise messages.
\end{proof}

\begin{lemma}\label{lemma:one11andall5}
 $p_0$ eventually executes Line 15; when this happen, all other processes have executed Line 6.
\end{lemma}

\begin{proof}
  Notice that while $p_0$ is waiting for a message on Line 15, it
  stops forwarding messages in the clockwise direction.  Therefore, by
  Lemma \ref{lemma:minid}, $p_0$ eventually receives a message on port
  $0$ since for the same lemma, $p_{1}$ eventually executes Line 6.
  Since $p_0$ is the only process that executes Line 14 (by Lemma
  \ref{lemma:minid}), there is a single counter-clockwise message in
  transit at any moment in the execution.

  Therefore, the only way for $p_0$ to receive a counter-clockwise
  message is to receive its own message back.  This happens only when
  all the other processes have received this message at Line 3 and
  forwarded it at Line 6.
\end{proof}


\begin{lemma}\label{lemma:declogging}
 The condition at Line 21 on $p_0$ is eventually true and when this happen: (1) there are no messages in transit on the network, and; (2) all other processes are executing the loop at Lines 8--11 and have received precisely one counter-clockwise message while executing it (at Line 9). 
\end{lemma}
\begin{proof}
  By Lemma \ref{lemma:minid} and Lemma \ref{lemma:one11andall5}, $p_0$
  executes Line 15 when all the other processes have already executed
  Line 6.  We have to show that no process $p_i \neq p_0$ can wait on
  Line 7 forever. Let $W$ be the set of processes waiting on Line
  7. Note that the processes in $W$ are CW-unbalanced while the others
  are CW-balanced. Note also that $W$ does not contain $p_0$.

  Once $p_0$ executes Line 16 it sends a counter-clockwise message
  $\msg_{cc}$. Notice also that, at this point, this is the only
  counter-clockwise message in transit. Therefore as long as $p_0$ has
  not executed Line 22 (and exited the loop at Lines 17--21), the
  condition at Line 11 cannot be satisfied for any process
  $p_i \neq p_0$.  Therefore, no process can terminate the execution
  of the algorithm before $\msg_{cc}$ gets back to $p_0$.  Is easy to see
  that $\msg_{cc}$ cannot be forwarded by a process in $W$, thus when
  $p_0$ receives $\msg_{cc}$ back, making the condition at Line 21 true;
  the set $W$ must be empty.

Notice that at the last iteration of the loop in Lines 1--4 each process $p_i \neq p_0$ sent a clockwise message  at Line 2 but received a counter-clockwise message at Line 3, thus there are at most $n-1$  clockwise messages in transit (as $p_0$ is CW-balanced after it exited the loop at Lines 1--4). Each $p_i \neq p_0$ will wait on Line 7 until it eliminates one of these messages, and then while executing the loop at Lines 8--11, it relays both $\msg_{cc}$ and the remaining clockwise messages. 
It is easy to see that eventually one process $p_i \in W$ will be removed from $W$, as there are $|W|$ messages circulating clockwise and all the processes that are not in $W$ are forwarding messages. Notice also that once $p_i$ is removed from $W$, it starts forwarding messages. This implies that another process has to be removed from $W$ and that eventually $W$ will be empty. When $W$ is empty all the $n-1$ clockwise messages have been eliminated from the network.  

Therefore $\msg_{cc}$ eventually reaches back $p_0$, and at this point
there are no more clockwise messages in the network. Since there is no
counter-clockwise message in transit at this point, this prove (1);
(2) follows immediately by observing that $\msg_{cc}$ has been forwarded
by every process executing the Loop at lines 8--11.
\end{proof}

\begin{lemma}
 $p_0$ eventually executes Line 24. When $p_0$ executes Line 24 there is no message in transit on the network all the other processes have executed Line 12. 
\end{lemma}

\begin{proof}
By Lemma \ref{lemma:declogging}, $p_0$ eventually executes Line 22 and when this happen the only message circulating in the network is the counter-clockwise message $\msg_{cc}$ just created at Line 22. Moreover, by the same lemma, all other processes are executing the loop at Lines 8--11 and have already received a counter-clockwise message while executing it. 
Message $\msg_{cc}$ will be relayed to $p_0$ along the oriented path $p_{n-1},p_{n-2},\ldots,p_{1}$. Once $p_j$ relays $\msg_{cc}$ its condition at Line 11 is true and it exits the relay loop executing the return at Line 12. Therefore, once $\msg_{cc}$ reaches $p_0$, allowing $p_0$ to return at Line 24, all other processes returned at Line 12 and no other message is in the network. 
\end{proof}

We can now show the correctness and message complexity of our algorithm: 
\algone*

\begin{proof}
The correctness of the Algorithm \ref{algo:c1} trivially follows from the previous lemmas. 
Regarding the message complexity. The fact that there are exactly three counter-clockwise messages follows from the fact that $p_0$ is the only one executing Lines 14--24 by Lemma \ref{lemma:minid}, recall that w.l.o.g. we assume $p_0$ to be the process with minimum identifier. 
The bound on clockwise messages comes from observing that by Lemma \ref{lemma:minid} once all processes exited the Loop at lines 1--4 they have sent at most $U \myID_0+n$ messages and there are at most $n-1$ clockwise messages in transit (observe that when  exiting the loop, $p_0$ is CW-balanced and all  other processes are CW-unbalanced). Each process $p_i \neq p_0$ eliminates one clockwise message at Line 7, therefore a process can relay at most $n-1$ clockwise messages either at Lines 9--10 or Lines 18--20. From this immediately follows that a process sends at most $U \myID_0+2n$ clockwise messages. From this the total message complexity immediately follows. 
\end{proof}

\color{black}


\section{Conclusion}\label{sec13}

In this paper we have investigated the message complexity of leader election in content-oblivious oriented rings. 
 
 We showed that there exists an algorithm that sends $O(n \cdot  U \cdot \myID_{min})$ messages in the clockwise direction and where each process sends 3 messages in the counter-clockwise direction, and another algorithm that sends $O(n \cdot  U \cdot  \log (\myID_{min}))$ messages clockwise and $O(n \cdot  \log (\myID_{min}))$ messages counter-clockwise. These results are interesting for two reasons: sending a constant number of messages in one direction requires the non-uniform assumption, as demonstrated by our impossibility result; and the logarithmic factor is necessary, as shown by \cite{frei2024content}.
Finally, our randomized algorithms improves on the message complexity of  previous algorithms while providing termination guarantees. 

Our study also leaves several open directions for future work.

The most pressing one is to close the gap between the lower bound of $\Omega\left(n \cdot \log \tfrac{\myID_{\max}}{n}\right)$ and the upper bound of $O(n \cdot U \cdot \log (\myID_{\min}))$ for leader election on oriented rings. With the current state of knowledge, it is unclear whether the multiplicative factor $U$ is necessary.

Another important question is whether one can design an algorithm with only logarithmic dependency on the identifiers for general two-edge-connected networks. At present, the best known result is the algorithm of \cite{cha2025content}, which has complexity $O(m \cdot U \cdot \myID_{\min})$.


 

\bibliographystyle{plainurl}
\bibliography{refs}

\begin{thebibliography}{10}

\bibitem{angluin1980local}
D.~Angluin.
\newblock Local and global properties in networks of processors (extended abstract).
\newblock In {\em {STOC} 1980}, pages 82--93. {ACM}, 1980.

\bibitem{attiya2004distributed}
H.~Attiya and J.~L. Welch.
\newblock {\em Distributed Computing: Fundamentals, Simulations, and Advanced Topics}.
\newblock Wiley, 2004.

\bibitem{BVelection}
P.~Boldi, S.~Shammah, S.~Vigna, B.~Codenotti, P.~Gemmell, and J.~Simon.
\newblock Symmetry breaking in anonymous networks: Characterizations.
\newblock In {\em {ISTCS} 1996}, pages 16--26, 1996.

\bibitem{boldi2002fibrations}
P.~Boldi and S.~Vigna.
\newblock Fibrations of graphs.
\newblock {\em Discrete Mathematics}, 243(1):21--66, 2002.

\bibitem{burns1980formal}
J.~E. Burns.
\newblock A formal model for message passing systems.
\newblock Technical Report~91, Indiana University, September 1980.

\bibitem{cachin2011introduction}
C.~Cachin, R.~Guerraoui, and L.~Rodrigues.
\newblock {\em Introduction to Reliable and Secure Distributed Programming}.
\newblock Springer, 2011.

\bibitem{CASTEIGTS201920}
A.~Casteigts, Y.~Métivier, J.M. Robson, and A.~Zemmari.
\newblock Counting in one-hop beeping networks.
\newblock {\em Theoretical Computer Science}, 780:20--28, 2019.

\bibitem{CASTEIGTS201932}
A.~Casteigts, Y.~Métivier, J.M. Robson, and A.~Zemmari.
\newblock Design patterns in beeping algorithms: Examples, emulation, and analysis.
\newblock {\em Information and Computation}, 264:32--51, 2019.

\bibitem{censor2023distributed}
K.~Censor-Hillel, S.~Cohen, R.~Gelles, and G.~Sela.
\newblock Distributed computations in fully-defective networks.
\newblock {\em Distributed Computing}, 36(4):501--528, 2023.

\bibitem{cha2025content}
J.~Chalopin, Y.-J. Chang, L.~Chen, G.~A. Di~Luna, and H.~Zhou.
\newblock Content-oblivious leader election in 2-edge-connected networks.
\newblock In {\em {DISC} 2025: Proceedings of the 39th International Symposium on Distributed Computing}, LIPIcs. Schloss Dagstuhl - Leibniz-Zentrum f{\"{u}}r Informatik, 2025.
\newblock (to appear).

\bibitem{chalopin2012election}
J.~Chalopin, E.~Godard, and Y.~M{\'{e}}tivier.
\newblock Election in partially anonymous networks with arbitrary knowledge in message passing systems.
\newblock {\em Distributed Computing}, 25(4):297--311, 2012.

\bibitem{cornejo2010deploying}
A.~Cornejo and F.~Kuhn.
\newblock Deploying wireless networks with beeps.
\newblock In {\em {DISC} 2010}, volume 6343 of {\em Lecture Notes in Comput. Sci.}, pages 148--162. Springer, 2010.

\bibitem{czumaj2019leader}
A.~Czumaj and P.~Davies.
\newblock Leader election in multi-hop radio networks.
\newblock {\em Theoretical Computer Science}, 792:2--11, 2019.

\bibitem{dufoulon2018}
F.~Dufoulon, J.~Burman, and J.~Beauquier.
\newblock Beeping a deterministic time-optimal leader election.
\newblock In {\em {DISC} 2018}, volume 121 of {\em LIPIcs}, pages 20:1--20:17. Schloss Dagstuhl - Leibniz-Zentrum f{\"{u}}r Informatik, 2018.

\bibitem{elruby1991linear}
M.~El-Ruby, J.~Kenevan, R.~Carlson, and K.~Khalil.
\newblock A linear algorithm for election in ring configuration networks.
\newblock In {\em {HICSS} 1991}, volume~1, pages 117--123, 1991.

\bibitem{flocchini2004sorting}
P.~Flocchini, E.~Kranakis, D.~Krizanc, F.~L. Luccio, and N.~Santoro.
\newblock Sorting and election in anonymous asynchronous rings.
\newblock {\em Journal of Parallel and Distributed Computing}, 64(2):254--265, 2004.

\bibitem{frederickson1987electing}
G.~N. Frederickson and N.~A. Lynch.
\newblock Electing a leader in a synchronous ring.
\newblock {\em Journal of the ACM}, 34(1):98--115, 1987.

\bibitem{frei2024content}
F.~Frei, R.~Gelles, A.~Ghazy, and A.~Nolin.
\newblock Content-oblivious leader election on rings.
\newblock In {\em {DISC} 2024}, volume 319, pages 26:1--26:20. Schloss Dagstuhl - Leibniz-Zentrum f{\"{u}}r Informatik, 2024.

\bibitem{forster2014deterministic}
K.-T. Förster, J.~Seidel, and R.~Wattenhofer.
\newblock Deterministic leader election in multi-hop beeping networks.
\newblock In {\em {DISC} 2014}, volume 8784 of {\em Lecture Notes in Comput. Sci.}, pages 212--226. Springer, 2014.

\bibitem{hirschberg1980decentralized}
D.~S. Hirschberg and J.~B. Sinclair.
\newblock Decentralized extrema-finding in circular configurations of processors.
\newblock {\em Communications of the ACM}, 23(11):627--628, 1980.

\bibitem{Itai90negative}
Alon Itai and Michael Rodeh.
\newblock Symmetry breaking in distributed networks.
\newblock {\em Inf. Comput.}, 88(1):60–87, July 1990.
\newblock \href {https://doi.org/10.1016/0890-5401(90)90004-2} {\path{doi:10.1016/0890-5401(90)90004-2}}.

\bibitem{korach1984lower}
E.~Korach, D.~Rotem, and J.~Pach.
\newblock Lower bounds for distributed maximum-finding algorithms.
\newblock {\em Journal of the ACM}, 31(4):905--918, 1984.

\bibitem{kutten2020singularly}
S.~Kutten, W.~K. Moses~Jr., G.~Pandurangan, and D.~Peleg.
\newblock Singularly optimal randomized leader election.
\newblock In {\em {DISC} 2020}, volume 179 of {\em LIPIcs}, pages 22:1--22:18, 2020.

\bibitem{raynal2018fault}
M.~Raynal.
\newblock {\em Fault-Tolerant Message-Passing Distributed Systems: An Algorithmic Approach}.
\newblock Springer, 2018.

\bibitem{santoro2007design}
N.~Santoro.
\newblock {\em Design and Analysis of Distributed Algorithms}.
\newblock John Wiley \& Sons, 2007.

\bibitem{santoro1989time}
N.~Santoro and P.~Widmayer.
\newblock Time is not a healer.
\newblock In {\em {STACS} 1989}, volume 349 of {\em Lecture Notes in Comput. Sci.}, pages 304--313. Springer, 1989.

\bibitem{VacusZiccardi2025}
R.~Vacus and I~Ziccardi.
\newblock Minimalist leader election under weak communication.
\newblock In {\em {PODC} 2025}. {ACM}, 2025.

\bibitem{YKsolvable}
M.~Yamashita and T.~Kameda.
\newblock Computing on anonymous networks: Part {I} - {C}haracterizing the solvable cases.
\newblock {\em {IEEE} Transactions on Parallel and Distributed Systems}, 7(1):69--89, 1996.

\end{thebibliography}


\appendix

\section{Non-Uniform Lower Bound Adapted From Prior Work} \label{sect: lower bound}

It was shown in~\cite[Section~5]{frei2024content} that, for any uniform algorithm, there exists an identifier assignment under which the algorithm must send $\Omega\left(n \log\frac{\myID_{\max}}{n}\right)$ messages on oriented rings. In this section, we briefly review their lower bound proof and show how to extend the lower bound to the \emph{non-uniform} setting, even when the \emph{exact} value of $n$ is known. 

For both the uniform and non-uniform settings, the lower bound is 
$\Omega\left(n \log \frac{k}{n}\right)$,
where $k$ denotes the number of assignable identifiers. By restricting the identifier space to $[k+1, 2k]$, we ensure that both $\myID_{\max}$ and $\myID_{\min}$ are $\Theta(k)$. Consequently, we obtain lower bounds not only of the form 
$\Omega\left(n \log \frac{\myID_{\max}}{n}\right)$ but also of the form 
$\Omega\left(n \log \frac{\myID_{\min}}{n}\right)$.

\subparagraph{Solitude Patterns.} A key technique underlying their proof is the notion of the \emph{solitude pattern}~\cite[Definition~21]{frei2024content} of an identifier, defined as follows.  Suppose $\mathcal{A}$ is any deterministic leader election algorithm that quiescently terminates over rings of size $n=1$.  
Consider a ring $C'$ of size $n=1$, and let $v$ be the sole node in the ring.  
Run $\mathcal{A}$ on $v$ under a scheduler that always prioritizes earlier pulses, breaking ties in favor of $\cw$ pulses.  
The {solitude pattern} $p_v$ of $v$ is defined as a binary sequence encoding the order in which pulses are delivered, where bit~1 represents a $\cw$ pulse and bit~0 represents a $\ccw$ pulse.  
Since $\mathcal{A}$ guarantees quiescent termination, $v$ must eventually output $\leader$, implying that the solitude pattern has finite length.

\subparagraph{Uniform Lower Bound.}
Suppose $\mathcal{A}$ can also elect a leader over rings of size $n=2$.  
Then, as shown by Frei~et~al.~\cite{frei2024content}, distinct identifiers must have distinct solitude patterns.  
For contradiction, assume there exist two nodes $u$ and $v$ with $\id_u \neq \id_v$ and $p_u = p_v$ on a ring of size $n=2$.  
With a suitable adversarial scheduler, the view of either node can be made indistinguishable from that of being alone on a ring of size $n=1$.  
In this case, each node would erroneously elect itself as $\leader$, yielding a contradiction.

For $k$ distinct identifiers, there must be $k$ distinct solitude patterns.  
Consider the first $\lfloor \log(k/n) \rfloor$ bits of these patterns.  
There are at most $2^{\lfloor \log(k/n) \rfloor} - 1 \leq (k/n) - 1$ solitude patterns of length shorter than $\lfloor \log(k/n) \rfloor$.  
Hence, at least $k - \left((k/n) - 1\right) = (n-1)(k/n) + 1$ solitude patterns must have length at least $\lfloor \log(k/n) \rfloor$.  

Since there are at most $k/n$ possible prefixes of length $\lfloor \log(k/n) \rfloor$, by the pigeonhole principle, there exist $n$ distinct identifiers whose solitude patterns share the same first $\lfloor \log(k/n) \rfloor$ bits.  
Assign these $n$ identifiers to the nodes of a ring of size $n$.  
Under a suitable adversarial scheduler, the view from each node can be made indistinguishable from that of being alone on a ring of size $1$.  
Consequently, no node terminates before sending $\lfloor \log(k/n) \rfloor$ pulses.  
This yields the message complexity lower bound $\Omega\left(n \log \frac{k}{n}\right)$, where $n$ is the ring size and $k$ is the number of assignable identifiers.

\subparagraph{Strengthening the Lower Bound to Non-Uniform Algorithms.} The above lower bound proof~\cite{frei2024content} relies on the assumption that $\mathcal{A}$ must correctly elect a leader on rings of size $n \in \{1,2\}$, and therefore does not extend to non-uniform algorithms.  
In the following discussion, we show how to modify the argument so that the same lower bound $\Omega\left(n \log \frac{k}{n}\right)$ holds even when the algorithm is given the \emph{exact} value of $n$ in advance.

\begin{proposition} [{Non-uniform extension of~\cite[Theorem~20]{frei2024content}}]
\label{prop: lower bound}
For any deterministic, quiescently terminating leader election algorithm $\mathcal{A}$ on an oriented ring $C$ of $n$ nodes, where the \ul{exact} value of $n$ is known {a priori}, the message complexity of $\mathcal{A}$ is $\Omega\left(n \log \frac{k}{n}\right)$, where $k$ denotes the number of assignable identifiers.
\end{proposition}

\begin{proof}
We continue to use the notion of solitude patterns, but with relaxed requirements on the algorithm $\mathcal{A}$.  
We no longer require $\mathcal{A}$ to \emph{quiescently terminate} or to \emph{correctly} elect a leader on a ring of size $n=1$.  
Consequently, the solitude pattern $p_v$ of a node $v$ may now be an infinite binary sequence.  
Moreover, we also drop the correctness requirement for rings of size $n=2$, so solitude patterns of different identifiers are \emph{not necessarily distinct}.

Based on the first $\lfloor \log(k/n) \rfloor - 2$ bits of the $k$ assignable identifiers, we classify these identifiers into equivalence classes as follows.  

\begin{description}
    \item[Type-1 Classes:] For every binary string $s_1$ of length \emph{less than} $\lfloor \log(k/n) \rfloor - 2$, define  
    \[\langle s_1 \rangle = \{\text{all assignable identifiers whose solitude pattern is $s_1$}\}.\]

    \item[Type-2 Classes:] For every binary string $s_2$ of length \emph{exactly} $\lfloor \log(k/n) \rfloor - 2$, define  
    \[\langle s_2 \rangle = \{\text{all assignable identifiers whose  solitude pattern begins with $s_2$}\}.\]
\end{description}

Every equivalence class is associated with a distinct binary string of length \emph{less than} $\lfloor \log(k/n) \rfloor - 1$, so the total number of equivalence classes is less than $\frac{k}{2n}$.  
Since there are $k$ assignable identifiers, it follows that at least one equivalence class $S$ must contain at least $2n$ identifiers. We divide the analysis into two cases.
    
    \subparagraph{Case 1: $S = \langle s_1\rangle$ is a Type-1 Equivalence Class.}
    Since $s_1$ is a finite string, every identifier $\id \in S= \langle s_1\rangle$ \emph{does} quiescently terminate on a singleton ring.  
Because correctness is not required on rings of size $1$, the output may be either $\leader$ or $\nonleader$.  
However, at most one identifier in $S$ can output $\leader$, for otherwise assigning two such identifiers to a ring of size $n$ would contradict the requirement that exactly one leader must be elected on rings of size $n$.  
Thus, at least $2n-1$ identifiers in $S$ must output $\nonleader$.  
However, if we select $n$ such identifiers that output $\nonleader$ and place them on a ring of size $n$, no leader would be elected, again contradicting the correctness of the algorithm.  
We conclude that algorithm $\mathcal{A}$ must be erroneous.

    \subparagraph{Case 2: $S = \langle s_2 \rangle$ is a Type-2 Equivalence Class.}  
The analysis in this case is similar to the lower bound argument for uniform algorithms discussed above.  
If we select any $n$ identifiers from the class $S = \langle s_2 \rangle$ and assign them to a ring of size $n$, an adversarial scheduler can ensure that no node terminates before sending at least $\lfloor \log(k/n) \rfloor - 2$ bits.  
This yields a message complexity lower bound of $\Omega\left(n \log \tfrac{k}{n}\right)$.
\end{proof}

\section{Non-Uniform Randomized Algorithms Adapted From Prior Work} 
\label{sect:adaptedalgo}

In anonymous systems, where processes lack identifiers, Frei~et~al.~\cite{frei2024content} presented a \emph{uniform} leader election algorithm that uses $n^{O(c^2)}$ messages and succeeds with probability $1 - n^{-c}$, but does not guarantee explicit termination. To obtain explicit termination, \emph{non-uniformity} is necessary, as Itai and Rodeh~\cite{Itai90negative} showed that no asynchronous, anonymous ring can decide its size using a \emph{uniform} and \emph{terminating} algorithm. This impossibility directly implies the non-existence of a uniform and terminating leader election in the asynchronous and anonymous setting, otherwise one could subsequently infer the ring size by the simulation result of~\cite{censor2023distributed}.

 The high message complexity $n^{O(c^2)}$ in the algorithm of Frei~et~al.~\cite{frei2024content} arises from enforcing uniformity together with the success probability guarantee of $1 - n^{-c}$. In this section, we briefly review their uniform randomized algorithm, explain how their approach can be adapted to obtain a non-uniform algorithm with improved message complexity, and then compare this adapted algorithm with our new randomized algorithm, which relies on a different technique.

There are two main deterministic results in~\cite{frei2024content}:

\begin{description}
    \item[Result 1:] For \emph{oriented} rings, there exists a \emph{quiescently terminating} leader election algorithm that uses $O(n \cdot \idmax)$ messages~\cite[Theorem~1]{frei2024content}. This algorithm requires all nodes to have \emph{distinct} identifiers. The explicit requirement of {distinct} identifiers appears in Lemma~22 of the full version of~\cite{frei2024content}.
    \item[Result 2:] For \emph{unoriented} rings, there exists a \emph{stabilizing} leader election algorithm that uses $O(n \cdot \idmax)$ messages~\cite[Theorem~2]{frei2024content}. In this case, the algorithm does not guarantee explicit termination. The algorithm  only requires that the node holding the maximum identifier is unique, and does not require all identifiers to be distinct.
\end{description}

Since in the anonymous setting nodes are not assigned $\id$s a priori, using the above deterministic algorithms as a black box in designing randomized algorithms requires generating $\id$s via randomness.

\subparagraph{Uniform Randomized Algorithms.} Using Result~2, Frei~et~al.~\cite{frei2024content} presented a \emph{uniform} leader election algorithm that uses $n^{O(c^2)}$ messages and succeeds with probability $1 - n^{-c}$, but does not guarantee explicit termination. Each node samples a value $\mathsf{BitCount}$ from a geometric distribution with parameter $1 - 2^{-\Theta(1/c)} = \Theta(1/c)$, so that $\mathbb{E}[\mathsf{BitCount}] = \Theta(c)$ and $\mathsf{BitCount} = O(c^2 \log n)$ with probability $1 - n^{-c}$. The node then samples its $\id$ uniformly at random from the set of identifiers of length at most $\mathsf{BitCount}$. With probability $1 - n^{-c}$, the maximum identifier satisfies $\idmax = n^{\Omega(c)}$ and also $\idmax = n^{O(c^2)}$. The lower bound guarantees that the maximum identifier is unique with probability $1 - n^{-c}$, while the upper bound ensures that the overall message complexity is $n^{O(c^2)}$.

We remark that the high message complexity $n^{O(c^2)}$ is partly due to enforcing a high success probability of $1 - n^{-c}$. If one is satisfied with a success probability $1 - 1/R$ for some \emph{constant} $R > 1$, then it suffices to sample $\id$s from a geometric distribution with mean $\Theta(R)$, which ensures that the maximum identifier is unique with probability $1 - 1/R$, yielding an overall success probability of $1 - 1/R$. With this modification, the overall message complexity improves to $O(n \log n)$, as $\idmax = O(R \log n)$ with high probability. This result can be compared with our randomized algorithm, which also achieves near-linear message complexity: Our algorithm further guarantees a high success probability of $1 - U^{-c} \geq 1 - n^{-c}$ and explicit termination, though at the cost of non-uniformity.

Due to the requirement of distinct identifiers, Result~1 cannot be applied in the uniform setting. In fact, when the sample size $n$ is large, no fixed discrete distribution $\mathcal{D}$ independent of $n$ can guarantee \emph{distinct} samples with constant probability. The argument is straightforward: let $X_1, X_2, \ldots, X_n$ be i.i.d.~random variables drawn from $\mathcal{D}$, and suppose $\Pr[X_1 = k] = p > 0$ for some $k$. By independence, $\Pr[X_2 = k] = p$ as well, so $\Pr[X_1 = X_2] = p^2$. The same holds for disjoint pairs $(X_3, X_4)$, $(X_5, X_6)$, and so on. For even $n$, this gives $\Pr\left[\bigwedge_{j=1}^{n/2} (X_{2j-1} \neq X_{2j})\right] = (1 - p^2)^{n/2}$. As $n$ grows, this probability approaches $0$, so for sufficiently large $n$, the probability of obtaining all distinct samples falls below any constant $\varepsilon > 0$. This shows that without an upper bound on $n$, Result~1 cannot be adapted to the randomized setting. Indeed, as discussed above, no quiescently terminating leader election can succeed with constant probability $\varepsilon$ in the anonymous setting; see Corollary~26 of the full version of~\cite{frei2024content} and \cite[Theorem~4.2]{Itai90negative}.

\subparagraph{Non-Uniform Randomized Algorithms.} For \emph{non-uniform} algorithms, the above obstacle can be circumvented. Given an upper bound $U$ on $n$, it suffices to ensure that for each random variable $X_i \sim \mathcal{D}$, the probability $\Pr[X_i = k]$ is at most $O(U^{-c})$ for every possible outcome $k$. This guarantees that the samples are \emph{distinct} with probability at least $1 - \binom{n}{2} U^{-c} \geq 1 - U^{2-c}$. Such a distribution can be realized by choosing identifiers uniformly at random from the set of all identifiers of length at most $\lceil c \log U \rceil$. In this way, one can directly apply the deterministic algorithm of Result~1. Since $\idmax = O(U^c)$, the overall message complexity is $O(n \cdot \idmax) = U^{O(c)}$. As with our randomized algorithm, this yields a non-uniform quiescently terminating leader election algorithm with a high success probability guarantee. However, the message complexity $O(U^c)$ of this approach is significantly higher than the $O(n \log^2 U)$ message complexity of our algorithm.

\end{document}